\documentclass[journal]{IEEEtran}

\IEEEoverridecommandlockouts

\usepackage{bm}
\usepackage{epsf}
\usepackage{subfigure}
\usepackage{cite}
\usepackage{graphics} 
\usepackage{epsfig} 
\usepackage{amsthm}
\usepackage{mathrsfs}
\usepackage{mathptmx} 
\usepackage{times} 
\usepackage{amsmath} 
\usepackage{amssymb}  
\usepackage{psfrag}
\usepackage{enumerate}
\usepackage{url}
\usepackage{stfloats}
\usepackage{array}
\usepackage{latexsym}
\usepackage{amssymb} 
\usepackage[usenames]{color}
\usepackage[ruled,vlined]{algorithm2e}
\usepackage{wrapfig}

\usepackage{booktabs}
\usepackage{threeparttable}

\newtheorem{definition}{Definition}
\newtheorem{theorem}{Theorem}

\newtheorem{proposition}{Proposition}
\newtheorem{lemma}{Lemma}

\newtheorem{remark}{Remark}
\newtheorem{example}{Example}
\newtheorem{claim}{Claim}

\newcommand{\lfteqn}{\begin{eqnarray} \begin{array}{lllllll}}
\newcommand{\ndeqn}{\end{array} \nonumber \end{eqnarray}}
\newcommand{\Lfteqn}{\begin{eqnarray} \begin{array}{lllllll}}
\newcommand{\Ndeqn}{\end{array}  \end{eqnarray}}

\begin{document}

\title{\LARGE \bf A New Approach for Verification of Delay Coobservability of Discrete-Event Systems
}

\author{Yunfeng~Hou, ~Qingdu~Li, ~Yunfeng~Ji,~\IEEEmembership{Member,~IEEE},  ~Gang~Wang ~\IEEEmembership{Member,~IEEE}, and  ~Ching-Yen~Weng

\thanks{Yunfeng Hou, Qingdu Li, Yunfeng Ji, and  Gang Wang are with the Institute of Machine Intelligence, University of Shanghai for
Science and Technology, Shanghai 200093, China. 
Ching-Yen Weng is with  the Robotics Research Centre, Nanyang Technological University, Singapore.
(Email: yunfenghou@usst.edu.cn; liqd@usst.edu.cn; ji$\_$yunfeng@usst.edu.cn; gwang@usst.edu.cn; weng0025@e.ntu.edu.sg).
}}



\maketitle

\begin{abstract}
In decentralized networked supervisory control of discrete-event systems (DESs), the local supervisors observe event occurrences subject to observation delays to make correct control
decisions. Delay coobservability describes whether these local
supervisors can make sufficient observations. 
In this paper, we provide an efficient way to verify delay coobservability.
For each controllable event, we partition the specification language into a finite number of sets such that strings in
different sets have different lengths.  
For each of the sets, we construct a verifier to check if delay coobservability holds for the controllable event.
The computational complexity of the proposed approach is polynomial with respect to the number of states, the number of events, and the upper bounds on observation delays and only exponential with respect to the number of local supervisors.
It has lower complexity order than the existing approaches.
In addition, we investigate the relationship between the decentralized supervisory control of networked DESs and the decentralized fault diagnosis of  networked DESs and show that delay $K$-codiagnosability is transformable to delay coobservability. 
Thus, techniques for the verification of delay coobservability can be leveraged to verify delay $K$-codiagnosability.

\end{abstract}


\section{Introduction}

\IEEEPARstart{I}{n} \textcolor{blue}{the context of DESs, observability has been a vital issue for supervisory control.
A system is said to be observable if the extensions of any two confusable strings (having the same natural projection) with the same controllable event should be
either both within the legal language or both out of it \cite{lin88is}.
Observability has been extensively investigated over the past decades. 
For example,  it has been used to solve the decentralized supervisory control problem  \cite{lin88isb, lin90tac,rudie92tac,rudie95tac,yoo02deds,yoo04tac,yin16tac3},  where the plant is controlled by a set of local supervisors.
It also has been used to solve a robust control problem, where the system is not entirely
known \cite{lin93tac}. 
Relevant properties stronger than observability, including normality \cite{lin88is}, weak normality \cite{takai02auto}, strong observability\cite{takai05tac}, and
relative observability \cite{cai15tac}, were also proposed.}

\textcolor{blue}{Advances in network technology enable us to connet the supervisor(s) with the plant using networks.
Such a networked information structure not only provides efficient ways for controlling DES, but also brings new challenges on achieving the control objective, e.g., how to overcome the delays and losses occurring in the communication between the plant and the supervisor(s). \cite{lin2014control,shu14ac,shu17tac2,shu17tac1,shu15tac,hou20tac,liu21tac, wang16tase,rashidinejad18,alves20tac,alves17cdc, zhao17tcns}.
Recently, delay coobservability was employed to solve the decentralized (nonblocking) supervisory control problem under communication delays \cite{xu20deds}. 
In the delay coobservability,  if an event  needs to be disabled after the occurrence of a string, then there exists at least one local supervisor which can do this with certainty even if there exist communication delays.
Thus, the verification of delay coobservability is a crucial step in the decentralized (nonblocking) supervisory control under communication delays.
With this as motivation, we study the verification of delay coobservability in this paper.}

The problem of the verification of delay coobservability was initially studied in \cite{xu20tcns}. 
Given that the observation delays for each local supervisor are fixed, the authors in \cite{xu20tcns} constructed a verifier to search all the confusable string collections.
Based on the constructed verifiers, one can check if there exists a confusable string collection leading to a violation of delay coobservability.
Since there are $\prod_{i=1}^n(N_{o,i}+1)$ combinations of fixed observation delays, the number of verifiers to be constructed  in \cite{xu20tcns}  is $\prod_{i=1}^n(N_{o,i}+1)$ where $n$ and $N_{o,i}$ denote the number of local supervisors and the upper bounds on observation delays for local supervisor $i$,  respectively.

\textcolor{blue}{In this paper, we propose a new approach for the verification of delay coobservability.
Instead of searching all the confusable string collections, we search only those confusable string collections that  may cause a violation of delay coobservability.
}
The verification is performed for one controllable event at a time.
For a given controllable event $\sigma$, let $N \le \max\{N_{o,1},\ldots,N_{o,n}\}$ be the maximum observation delays for those supervisors for which the occurrence of $\sigma$ is controllable.
We first partition the specification language into  $N+1$ mutually disjoint sets  such that (i) the $k$th set, $k\in \{1,\ldots,N\}$, contains all the strings in the specification language with the length of $k-1$; and (ii) the $N+1$st set contains all the strings in the specification language with the length no smaller than $N$.
Then, for each of the partitioned sets, we construct a verifier to check whether all the strings in the set are safe with respect to (w.r.t.) the controllable event $\sigma$.
Therefore, the number of verifiers proposed in this paper is upper-bounded by  $(\max\{N_{o,1},\ldots,N_{o,n}\}+1)\times |\Sigma|$,  where $\Sigma$ is the event set of the plant. 
Computational complexity analysis shows that our approach has a worst-case complexity of $(n+1) \times |Q_H|^{n+1}\times |\Sigma|^2  \times \sum_{k=0}^{\max\{N_{o,1},\ldots,N_{o,n}\}} (k+1)$, where $Q_H$ is the state space of the plant.
This is in contrast to the complexity of $(n+1)\times |Q_H|^{n+1}\times|\Sigma|^{\max\{N_{o,1},\ldots,N_{o,n}\}+1}\times\prod_{i=1}^n(N_{o,i}+1)$ for the previous approach.
Thus, the worst-case complexity for our approach is polynomial in $|Q_H|$, $|\Sigma|$, and $\max\{N_{o,1},\ldots,N_{o,n}\}$ and exponential in $n$, whereas for the previous approach it is polynomial in $|Q_H|$ and $|\Sigma|$ and exponential in $n$ and $\max\{N_{o,1},\ldots,N_{o,n}\}$.

We also consider the decentralized fault diagnosis problem under the framework of networked DESs proposed in \cite{lin2014control,shu14ac, xu20deds}, where multiple local diagnosing agents work as a group  with possible observation delays to diagnose the system.
The notions of delay $K$-codiagnosability is introduced to capture whether the local agents can detect an unobservable fault event occurrence within $K$ steps, when observation delays exist.
We investigate the relationship between  delay coobservability and  delay $K$-codiagnosability and show that delay $K$-codiagnosability can be transformed into delay coobservability.
Thus, algorithms for verifying delay coobservability can be extended to verify delay $K$-codiagnosability.  
The difference between our work and the existing works will be discussed in Section VI.

The rest of this paper is organized as follows. 
Section II presents some preliminary concepts and reviews the definition of delay coobservability. 
A new approach for verifying delay coobservability is presented in Section III.
Section IV analyzes the computational complexity of the proposed approach and compares it with that of the existing approach.
Section V shows the application of the results derived in this paper.
Section VI studies the decentralized fault diagnosis problem of networked DESs and shows how to apply the approach for the verification of delay coobservability to verify delay $K$-codiagnosability.
Section VII concludes this paper.

\section{Preliminaries}
\label{Prelim}

\subsection{Preliminaries}

A deterministic finite automaton $G=(Q,\Sigma,\delta,\Gamma,q_0,Q_m)$ is used to describe a DES, where $Q$ is the finite set of states;
$\Sigma$ is the finite set of  events;
$q_0$ is the initial state;
$\delta:Q\times\Sigma\rightarrow Q$ is the transition function.
For a state $q\in Q$, active event set $\Gamma(q)\subseteq \Sigma$ is the set of events $\sigma$ such that $\delta(q,\sigma)$ is defined.
$Q_m$ is the set of marked states.
It should be noted that the initial state $q_0$ is not required to be a single state and can be a set of  states.
$\Sigma^*$ is the set of all strings that are composed of events in $\Sigma$.
$\delta$ can be iteratively extended to $Q\times\Sigma^*$ in the usual way \cite{lafortune07book}.
The language generated by $G$ is denoted by $\mathcal{L}(G)$, and the language marked by $G$ is $\mathcal{L}_m(G)$.
$\varepsilon$ is the empty string.

$\mathbb{N}$ is the set of natural numbers.
Given a natural number $M \in \mathbb{N}$, $\Sigma^{\le M}$ is the set of all strings in $\Sigma^*$ with a length no larger than $M$.
Let $[0,M]$ be the set of natural numbers no larger than $M$.
The prefix-closure of a string $s\in \mathcal{L}(G)$ is defined by $\overline{\{s\}}=\{u\in \Sigma^*:(\exists v\in \Sigma^*)uv=s\}$.
Given $L\in \Sigma^*$, $\overline{L}=\{u\in \Sigma^*:(\exists s\in L)u \in \overline{\{s\}}\}$.
$L$ is said to be prefix-closed if $L=\overline{L}$.
{$L$ is $\mathcal{L}_m(G)$-closed if $L=\overline{L} \cap \mathcal{L}_m(G)$.}
$G$ is said to be nonblocking if $\mathcal{L}(G)=\overline{\mathcal{L}_m(G)}$.

Given a string $s$, we denote its length by $|s|$.
Let $s_{-i}$ be the string in $\overline{\{s\}}$ with $|s_{-i}| = \max\{0, |s|-i\}$.
Let $s/s_{-i}$ be the suffix of $s$ such that $s=s_{-i}(s/s_{-i})$.
Given a string $s=\sigma_1\sigma_2\cdots\sigma_k\in \Sigma^*$, we write $s^i=\sigma_1\sigma_2\cdots\sigma_i$ for $i=1,2,\ldots,k$, and $s^0=\varepsilon$.
Given $G_1$ and $G_2$, the parallel composition of $G_1$ and $G_2$ is denoted by $G_1||G_2$ \cite{lafortune07book}.
We say $G_1$ is a sub-automaton  of $G_2$, denoted by $G_1 \sqsubseteq G_2$, if $G_1$ can be
obtained from $G_2$ by deleting some states in $G_2$ and all the transitions connected to these states.
The cardinality of a set $Z$ is denoted by $|Z|$.

In the context of decentralized networked supervisory control, $n$ local supervisors are involved.
We denote by $I=\{1,2,\ldots,n\}$ the index set of the local supervisors.
For each local supervisor $i\in I$, we denote the set of observable events as $\Sigma_{o,i} \subseteq \Sigma_o$ and the set of controllable events as $\Sigma_{c,i} \subseteq \Sigma_c$, where $\Sigma_o=\cup_{i=1}^n\Sigma_{o,i}$ and $\Sigma_c=\cup_{i=1}^n\Sigma_{c,i}$.
$\Sigma_{uo}=\Sigma\setminus\Sigma_o$ is the set of events that are unobservable to all local supervisors, and $\Sigma_{uc}=\Sigma\setminus \Sigma_c$ is the set of events that are uncontrollable to all local supervisors.
For a controllable event $\sigma \in \Sigma_c$, we denote by $I^c(\sigma)=\{i\in I:\sigma \in \Sigma_{c,i}\}$ the set of local supervisors that the event occurrence of $\sigma$ is controllable to them.
For local supervisor $i\in I$, the natural projection $P_i: \mathscr{L}(G) \rightarrow \Sigma_{o,i}^{\ast}$ is defined as $P_i(\varepsilon)=\varepsilon$ and, for all $s, s\sigma \in\mathscr{L}(G)$, $P_i(s\sigma)=P_i(s)\sigma$ if $\sigma \in \Sigma_{o,i}$, and  $P_i(s\sigma)=P_i(s)$, otherwise.
The inverse mapping $P_i^{-1}:\Sigma_{o,i}^*\rightarrow 2^{\Sigma^*}$ of $P_i$ is defined as follows: for all $t\in \Sigma_{o,i}^*$, $P_i^{-1}(t)=\{s\in\Sigma^*:P_i(s)=t\}$.
$P_i$ and $P_i^{-1}$ are extended from a string to a language in the usual way.

\subsection{Decentralized nonblocking networked supervisory control }
\label{Nobs}

The goal of decentralized nonblocking networked supervisory control  in \cite{xu20deds,xu20tcns} is to find a set of local supervisors that work as a group to achieve the specification language deterministically under control delays and observation delays.
Each local supervisor is connected to the plant via an independent observation channel.
The control commands issued by each local supervisor are delivered to the fusion site via an independent control channel.
The  adopted decentralized control architecture is the conjunctive and permissive architecture.
The protocol under this architecture can be described as follows: each local supervisor sends its event enablement commands to the fusion site, and control actions are then obtained by taking the intersection of these enabled events.

Due to the network characteristics, delays exist for both control and observation.
The assumptions made in this paper are as follows. 
  \begin{enumerate}
    \item{For each local supervisor $i\in I$, delays that occur in the observation channel $i$ (control channel $i$) are random but upper bounded by $N_{o,i}$ ($N_{c,i}$) event occurrences;}
    \item{For each observation channel, delays do not change the order of observations, i.e., first-in-first-out (FIFO) is satisfied;}
    \item{For each local supervisor $i\in I$, the control command being in effect  is the one that has most-recently sent to the fusion site;}
    \item{The initial control commands sent by all the local supervisors can be executed without any delays.}
    \end{enumerate}

By assumption 1), an event occurrence that is observable to the local supervisor $i$ can be delivered to local supervisor $i$ before no more than $N_{o,i}$ additional event occurrences, and a control command issued by the local supervisor $i$ can be delivered to the fusion site  before no more than $N_{c,i}$ additional event occurrences.
Due to control delays, for an occurred string $s\in \mathcal{L}(G)$ and a supervisor $i\in I$, the control commands being in effect at the fusion site can be any one of the control commands issued after the occurrence of $s_{-m_i}$, $m_i\in[0,N_{c,i}]$.
Due to observation delays, for an occurred string $s\in \mathcal{L}(G)$, what the supervisor $i$ may see is nondeterministic and denoted by $\Theta_i^{N_{o,i}}(s)=\{P_i(s_{-m_i}):m_i\in[0,N_{o,i}]\}$.
The inverse mapping of ${\Theta_i^{N_{o,i}}}$ is defined as follows: for all $t_i\in\Sigma_{o,i}^*$, $({\Theta_i^{N_{o,i}}})^{-1}(t_i)=\{s_i\in\mathcal{L}(G):t_i\in \Theta_i^{N_{o,i}}(s_i)\}$.
${\Theta_i^{N_{o,i}}}$ and $({\Theta_i^{N_{o,i}}})^{-1}$ are extended from a string to a language in the usual way.
The following conclusion is proven in \cite{xu20tcns}.
\begin{lemma}\label{Lem1}
For any $t \in \Theta_i^{N_{o,i}}(\mathcal{L}(G))$, the inverse mapping $({\Theta_i^{N_{o,i}}})^{-1}(t)$ is equal to the concatenation of the inverse mapping $P_i^{-1}(t)$ and the set of strings with lengths no larger than $N_{o,i}$. More precisely, $({\Theta_i^{N_{o,i}}})^{-1}(t)=P_i^{-1}(t)\Sigma^{\le N_{o,i}}.$
\end{lemma}

The desired system is represented in this paper by a sub-automaton of $G$ denoted $H=(Q_H,\Sigma,\delta_H,\Gamma_H,q_0,Q_{m,H})\sqsubseteq G$. 
We call $\mathcal{L}(H)$ as the specification language.
$\mathcal{L}(H)$ is controllable \cite{ramadge1987supervisory} w.r.t. $\Sigma_{uc}$ and $\mathscr{L}(G)$ if $\mathcal{L}(H)\Sigma_{uc} \cap \mathscr{L}(G) \subseteq \mathcal{L}(H)$.
The notion of delay coobservability is introduced to describe whether these local supervisors can make sufficient observations so that the correct control decisions can be made even if observation delays and control delays exist.
Formally, delay coobservability is defined in \cite{xu20deds} as follows.
\begin{definition}\label{Def1}
$\mathcal{L}(H)$ is delay coobservable w.r.t. $N_{o,1},\ldots, N_{o,n}$, and $\mathcal{L}(G)$, if for any string $s \in \mathcal{L}(H)$ and any controllable event $\sigma \in \Sigma_c$,
\begin{align}\label{Eq1}
&s\sigma\in \mathcal{L}(G) \wedge s\sigma \notin \mathcal{L}(H) \Rightarrow \notag \\
&(\exists i \in I^c(\sigma))({\Theta_i^{N_{o,i}}})^{-1}(\Theta_i^{N_{o,i}}(s))\sigma \cap \mathcal{L}(H)=\emptyset.
\end{align}
\end{definition}
By (\ref{Eq1}), if the occurrence of $\sigma$ needs to be disabled after $s$, then there exists at least one local supervisor who can disable the event occurrence of $\sigma$ and can distinguish $s$ from all the strings after which the event occurrence of $\sigma$ needs to be enabled, subject to observation delays.
It is shown in \cite{xu20tcns} that the decentralized nonblocking networked supervisory control problem is solvable if and only if (i)  the specification language $\mathcal{L}(H)$ is controllable w.r.t. $\Sigma_{uc}$ and $\mathscr{L}(G)$, (ii) $\mathcal{L}_m(G)$-closed, and (iii) delay coobservable w.r.t. $N_{o,1}+N_{c,1},\ldots, N_{o,n}+N_{c,n}$, and $\mathcal{L}(G)$.

Recall from  \cite{xu20tcns} that an augmented automaton  $H_{N_{o,i}}^{aug}=(Q_H \cup\{q_{dis}\},\Sigma,\delta_{H,N_{o,i}}^{aug},\Gamma_{H,{N_{o,i}}}^{aug}, q_0, Q_{H,m})$ is constructed as follows: for all $q \in Q_H$ and all $\sigma \in \Sigma$,
\lfteqn
\delta_{H,{N_{o,i}}}^{aug}(q,\sigma)=
\begin{cases}
	\delta_H(q,\sigma) & \mathrm{if}\ \sigma \in \Gamma_H(q)\\
	q_{dis} & \mathrm{if}\ \sigma \notin \Gamma_H(q)\wedge (\exists s'\in \Sigma^{\le {N_{o,i}}})\\&\sigma \in \Gamma_H(\delta_H(q,s'))\\
	\mathrm{undefined} & \mathrm{otherwise}.
\end{cases}
\ndeqn

In $H_{N_{o,i}}^{aug}$,  events that are active when the system is in state $q \in Q_H$ are those that are defined at states in $H$ can be reached from $q$ within $N_{o,i}$ steps. This gives us Lemma \ref{Lem2}.

\begin{lemma}\label{Lem2}
For any string $s\in \mathcal{L}(H)$ and any event $\sigma \in \Sigma$, the following statement is true.
\begin{align}
s\sigma \in \mathcal{L}(H_{N_{o,i}}^{aug})\Leftrightarrow (\exists s'\in \Sigma^{\le {N_{o,i}}})ss'\sigma \in \mathcal{L}(H).
\end{align}
\end{lemma}

Lemma \ref{Lem2} is proven in \cite{xu20tcns}.
We use an example to illustrate how $H_{N_{o,i}}^{aug}$ is constructed.

\begin{figure}
\centering \subfigure[Automaton
$G$]{\label{Fig11}\includegraphics[width=3.3cm]{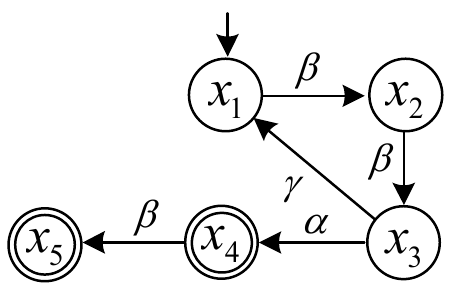}}
\subfigure[Automaton
$H$]{\label{Fig12}\includegraphics[width=2.1cm]{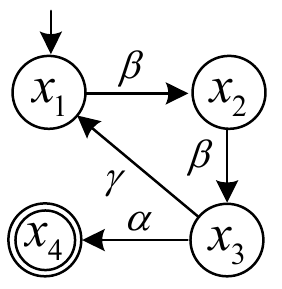}}
\caption{Uncontrolled system
$G$ and desired system $H$ in Example
\ref{Ex1}.} \label{Fig1}
\end{figure}

\begin{figure}
\centering \subfigure[Automaton
$H_{N_{o,1}}^{aug}$]{\label{Fig21}\includegraphics[width=3.4cm]{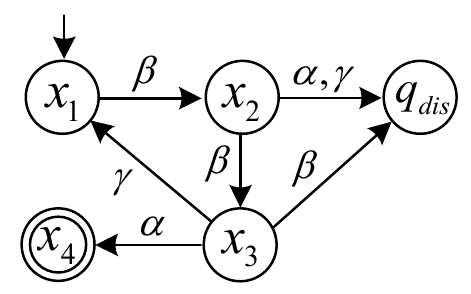}}
\subfigure[Automaton
$H_{N_{o,2}}^{aug}$]{\label{Fig22}\includegraphics[width=3.1cm]{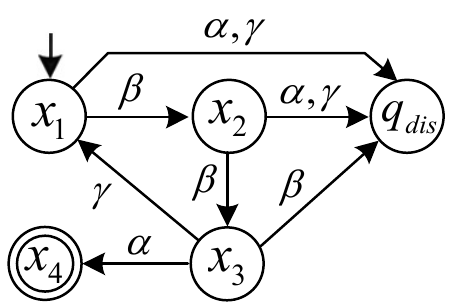}}
\caption{Automata $H_{N_{o,1}}^{aug}$ and $H_{N_{o,2}}^{aug}$ in Example
\ref{Ex1}.} \label{Fig2}
\end{figure}

\begin{example}\label{Ex1}
Consider system $G$ and the desired system $H$ that are depicted in Fig.\ref{Fig11} and Fig.\ref{Fig12}, respectively.
Suppose that there are two local supervisors, i.e., $I=\{1,2\}$.
Let $N_{o,1}=1$ and $N_{o,2}=2$.
The corresponding augmented automata $H^{aug}_{N_{o,1}}$ and $H^{aug}_{N_{o,2}}$ are shown in Fig.\ref{Fig21} and Fig.\ref{Fig22}, respectively.

Let us take the initial state $x_1$ in $H^{aug}_{N_{o,1}}$ and $H^{aug}_{N_{o,2}}$ as an example.
By Fig.\ref{Fig12}, system $H$ can reach states $x_1,x_2$ from state $x_1$ via the string with a length no larger than $1$. 
Since $\beta$ is defined at state $x_1$ and only event $\beta$ is active at state $x_2$, no additional transitions are defined at state $x_1$ in $H_{N_{o,1}}^{aug}$.
On the other hand, system $H$ can reach states $x_1,x_2,x_3$ from state $x_1$ via a sring with a length no larger than $2$. 
Since both $\alpha$ and $\gamma$ are active in state $x_3$ in $H$ and $\alpha,\gamma$ are not defined at state $x_1$, new transitions $\delta^{aug}_{H,N_{o,2}}(x_1,\alpha)=\delta^{aug}_{H,N_{o,2}}(x_1,\gamma)=q_{dis}$ are added into $H_{N_{o,2}}^{aug}$ in Fig.\ref{Fig22}.
\end{example}

\section{{Verification of delay coobservability}}
\label{sec3}

In this section, we consider the verification of  delay coobservability.
By definition, delay coobservability can be verified for each $\sigma \in \Sigma_c$, one by one. 
The system is delay coobservable iff it is delay coobservable for every $\sigma \in \Sigma_c$.
Therefore, without loss of generality (w.l.o.g.), we only present the verification for one controllable event, denoted by $\sigma\in \Sigma_c$.
The proposed procedures can be easily applied to the verification of other controllable events.

For convenience, we assume that $I^c(\sigma)=\{1,\ldots, l\}\subseteq I$.
Let $s$ be a string in $\mathcal{L}(H)$.
The set of confusable string collections associated with $s\in \mathcal{L}(H)$ and $\sigma \in \Sigma_c$ is denoted by:
$$\mathcal{T}_{conf}^{\sigma}(s)=\{(s_1,\ldots, s_l):(\forall i\in I^c(\sigma))P_{i}(s_i)\in \Theta_{i}^{N_{o,i}}(s)\}.$$

Since $\Theta_{i}^{N_{o,i}}(s)$ collects all the possible delayed observations for $s$, the local supervisor $i$ may see $P_{i}(s_i)$ when $s$ occurs in $G$.
Delay coobservability can then be characterized by using the definition of $\mathcal{T}^{\sigma}_{conf}(\cdot)$ as follows.

\begin{theorem}\label{Theo1}
$\mathcal{L}(H)$ is delay coobservable w.r.t. ${N_{o,1}},\ldots,$ ${N_{o,n}}$, $\sigma \in \Sigma_c$, and $\mathcal{L}(G)$ iff the following condition is true:
\begin{align}\label{Eq3}
&(\forall s \in \mathcal{L}(H))s\sigma \in \mathcal{L}(G)\setminus \mathcal{L}(H)
\Rightarrow  \notag \\
&(\forall (s_1,\ldots, s_l)\in \mathcal{T}^{\sigma}_{conf}(s)) (\exists i \in I^c(\sigma))s_i\sigma \notin \mathcal{L}(H_{N_{o,i}}^{aug}).
\end{align}
\end{theorem}

\begin{proof}
($\Leftarrow$) The proof is by contradiction.
Suppose that $\mathcal{L}(H)$ is not delay coobservable, i.e.,
$(\exists s \in \mathcal{L}(H))[s\sigma\in \mathcal{L}(G) \setminus \mathcal{L}(H)] \wedge [(\forall i \in I^c(\sigma))({\Theta_i^{N_{o,i}}})^{-1}(\Theta_i^{N_{o,i}}(s))\sigma \cap \mathcal{L}(H)\neq\emptyset]$.
Since $(\forall i\in I^c(\sigma))({\Theta_i^{N_{o,i}}})^{-1}(\Theta_i^{N_{o,i}}(s))\sigma \cap \mathcal{L}(H)\neq \emptyset$, by Lemma 1, $(\forall i\in I^c(\sigma))P^{-1}_i(\Theta_i^{N_{o,i}}(s))\Sigma^{\le N_{o,i}}\sigma \cap \mathcal{L}(H) \neq \emptyset$.
By Lemma 2, $(\forall i\in I^c(\sigma))P^{-1}_i(\Theta_i^{N_{o,i}}(s))\sigma \cap \mathcal{L}(H_{N_{o,i}}^{aug})\neq \emptyset$.
Hence, $\exists s_i \in \mathcal{L}(H)$ such that  $P_i(s_i)\in \Theta_i^{N_{o,i}}(s)$ and $s_i\sigma \in \mathcal{L}(H_{N_{o,i}}^{aug})$ for  $i\in I^c(\sigma)$.
Since $I^c(\sigma)=\{1,\ldots,l\}$, $(s_{1},\ldots,s_{l}) \in \mathcal{T}^{\sigma}_{conf}(s)$.
Moreover, since $s\sigma \in \mathcal{L}(G)\setminus \mathcal{L}(H)$ and $s_i\sigma \in \mathcal{L}(H_{N_{o,i}}^{aug})$, $i\in I^c(\sigma)$, we know that (\ref{Eq3}) does not hold.

($\Rightarrow$) Also by contradiction. Suppose that (\ref{Eq3}) does not hold, i.e.,
$(\exists s\in \mathcal{L}(H))[s\sigma\in \mathcal{L}(G) \setminus \mathcal{L}(H)] \wedge [(\exists (s_1,\ldots,s_l) \in \mathcal{T}_{conf}^{\sigma}(s))(\forall i \in I^{c}(\sigma))s_i\sigma \in \mathcal{L}(H_{N_{o,i}}^{aug})].
$
Since $(s_1,\ldots,s_l) \in \mathcal{T}^{\sigma}_{conf}(s)$, by definition,  $P_i(s_i)\in \Theta_i^{N_{o,i}}(s)$ for all $i\in I^c(\sigma)$.
Moreover, since $s_i\sigma \in \mathcal{L}(H_{N_{o,i}}^{aug})$, $P^{-1}_i(\Theta_i^{N_{o,i}}(s))\sigma \cap \mathcal{L}(H_{N_{o,i}}^{aug})\neq \emptyset$  for all $i\in I^c(\sigma)$.
By Lemma 2,  $P^{-1}_i(\Theta_i^{N_{o,i}}(s))\Sigma^{\le N_{o,i}}\sigma \cap \mathcal{L}(H) \neq \emptyset$ for all $i\in I^c(\sigma)$.
By Lemma 1, $({\Theta_i^{N_{o,i}}})^{-1}(\Theta_i^{N_{o,i}}(s))\sigma \cap \mathcal{L}(H)\neq \emptyset$ for all $i\in I^c(\sigma)$, which contradicts that $\mathcal{L}(H)$ is delay coobservable. 
\end{proof}

By Theorem \ref{Theo1}, to verify delay coobservability, it suffices to check if there exist $s \in \mathcal{L}(H)$ and $(s_1,\ldots, s_l)\in \mathcal{T}_{conf}^{\sigma}(s)$ causing a violation of $(\ref{Eq3})$.
Before we formally show how to do this, we first partition $\mathcal{L}(H)$ into $$\mathcal{L}(H)= \mathcal{L}^0(H) \dot{\cup} \mathcal{L}^1(H) \dot{\cup} \cdots  \dot{\cup}\mathcal{L}^{N}(H),$$ where $N=\max\{N_{o,1},\ldots, N_{o,l}\}$, such that $\mathcal{L}^k(H)=\{s \in \mathcal{L}(H):|s|=k\}$, $k=0,1,\ldots,N-1$ is the set of strings in $\mathcal{L}(H)$ with the length of $k$, and $\mathcal{L}^{N}(H)=\{s \in \mathcal{L}(H):|s| \ge N\}$ is the set of strings in $\mathcal{L}(H)$ with the length no smaller than $N$.
By the partition, we can verify whether (\ref{Eq3}) is true for strings in $\mathcal{L}(H)$ by checking whether (\ref{Eq3}) is true for strings in $\mathcal{L}^0(H),\mathcal{L}^1(H),\ldots,\mathcal{L}^N(H)$, one by one.

For any  $k\in \{0,1,\ldots,N\}$, we now construct a verifier $V^{k}_{\sigma}$ to check if there exist  $s \in \mathcal{L}^k(H)$ and  $(s_1,\ldots, s_l) \in \mathcal{T}_{conf}^{\sigma}(s)$  leading to a violation of (\ref{Eq3}).
Note that, the construction of $V_{\sigma}^N$ is slightly different from that of $V_{\sigma}^k$, $k=0,1,\ldots,N-1$ because strings in $\mathcal{L}^k(H)$, $k=0,1,\ldots,N-1$ are of the same length, but strings in $\mathcal{L}^N(H)$ are not.

Formally, we first construct 
$
V_{\sigma}^k=(X_{\sigma}^k,\tilde{\Sigma},f_{\sigma}^k, x_{\sigma,0}^k)
$
for $k=0,1,\ldots,N-1$, where
\begin{itemize}
\item{
the state space is $X_{\sigma}^k \subseteq \underbrace{Q_{H} \times  \cdots \times Q_{H}}_{l+1} \times [0,k]$;}
\item{
$\tilde{\Sigma} \subseteq \underbrace{(\Sigma \cup \{\varepsilon\}) \times  \cdots \times (\Sigma \cup \{\varepsilon\})}_{l+1} \setminus  \{(\underbrace{\varepsilon, \ldots, \varepsilon}_{l+1})\}$ is the event set;}
\item{
$x_{\sigma,0}^k=(\underbrace{q_{0}, \ldots, q_{0}}_{l+1},0)$ is the initial state;}
\item{
the transition function $f_{\sigma}^k: X_{\sigma}^k \times
\tilde{\Sigma} \rightarrow X_{\sigma}^k$ is defined as follows.
For each state $\tilde{q}=(q,q_1,\ldots,q_l,d)\in X_{\sigma}^k$ and each event $e \in \Sigma$, we need to consider the following five cases for each $i\in I^c(\sigma)$: 

{$\mathrm{C}_1$: $k-d >N_{o,i}$ and $e \in \Sigma_{o,i}$;}

{$\mathrm{C}_2$: $k-d >N_{o,i}$ and $e \in \Sigma_{uo,i}$;} 

{$\mathrm{C}_3$: $k-d \le N_{o,i}$ and $\sigma \notin \Gamma_{H,N_{o,i}}^{aug}(q_i)$ and $e \in  \Sigma_{o,i}$;}

{$\mathrm{C}_4$: $k-d \le N_{o,i}$ and $\sigma \notin \Gamma_{H,N_{o,i}}^{aug}(q_i)$ and $e \in  \Sigma_{uo,i}$;}

{$\mathrm{C}_5$: $k-d \le N_{o,i}$ and $\sigma \in \Gamma_{H,N_{o,i}}^{aug}(q_i)$.}

Then, the following two types of transitions are defined in $V_{\sigma}^k$:

\begin{enumerate}
\item if $d+1 \le k$, $\delta_H(q,e)$ is defined, and for each $i \in I^c(\sigma)$, $\delta_H(q_i,e)$ is defined for $\mathrm{C}_1$ or $\mathrm{C}_3$. Then, a first type of transition is defined at $\tilde{q}$ as:
\begin{align}\label{Eq5-1}
&f_{\sigma}^k((q,q_1,\ldots,q_l,d),(e,e_1,\ldots,e_l))= \notag \\
&(\delta_H(q,e),\delta_H(q_1,e_1),\ldots,\delta_H(q_l,e_l),d+1),
\end{align}
where, for all $i \in I^c(\sigma)$,
\lfteqn
e_i= \begin{cases}
e & \text{if}\ \mathrm{C}_1\ \text{or}\ \mathrm{C}_3 \\
\varepsilon & \text{if}\ \mathrm{C}_2 \ \text{or}\ \mathrm{C}_4 \ \text{or}\ \mathrm{C}_5.
\end{cases}
\ndeqn

\item for each $i\in I^c(\sigma)$, if $\delta_H(q_i,e)$ is defined for  $\mathrm{C}_2$ or $\mathrm{C}_4$, then a second type of transition is defined at $\tilde{q}$ as:
\begin{align}\label{Eq5-2}
&f_{\sigma}^k((q,q_1,\ldots,q_l,d),(\varepsilon,\varepsilon,\ldots, \varepsilon, \underset{(i+1)^{st}}{e}, \varepsilon, \ldots,\varepsilon))= \notag\\
&(q, q_1,\ldots,q_{i-1}, \delta_H(q_{i},e), q_{i+1}, \ldots, q_{l},d).
\end{align}
\end{enumerate}

}\end{itemize}

The construction of $V^k_{\sigma}$ is briefly summarized as follows. 
For any state $\tilde{q}=(q,q_1,\ldots,q_l,d)$ in $V_{\sigma}^k$, the first component $q$ tracks the string that has occurred in the system and the integer $d$ records the length of this string.
Therefore, in $\mathrm{C}_1 \sim \mathrm{C}_5$, $k-d$ is the distance of the first component between the current state  $\tilde{q}=(q,q_1,\ldots,q_l,d)$ and  $\tilde{q}'=(q',q'_1,\ldots,q'_l,d')$ with $d'=k$.

If an event occurs in the system, i.e., the first component of the event defined at $\tilde{q}=(q,q_1,\ldots,q_l,d)$ is not $\varepsilon$ but some $e \in \Sigma$,  then for all $i \in I^c(\sigma)$ satisfying $\mathrm{C}_1$ or $\mathrm{C}_3$, the $i+1$st component of $\tilde{q}$, i.e., $q_i$ shall move together with the system to match the observation of $e$.
As shown in (\ref{Eq5-1}), for all $i \in I^c(\sigma)$ satisfying $\mathrm{C}_1$ or $\mathrm{C}_3$, we have $e_i=e$.
However, for all $i \in I^c(\sigma)$ satisfying $\mathrm{C}_5$, we keep $q_i$ unchanged in the successor states of $\tilde{q}$ (if they exist) no matter whether the occurrence of $e$ is observable to supervisor $i$ or not.
Additionally, when an event occurs in the system, $d$ is updated to $d+1$ to record the length of the system string.
Since $V_{\sigma}^k$ is used to consider all the strings in $\mathcal{L}^k(H)$, a first type of transition is only defined at a state $(q,q_1,\ldots,q_l,d)\in X^k_{\sigma}$ with $d+1 \le k$.

\textcolor{blue}{If no event occurs in the system, i.e., the first component of the event defined at $\tilde{q}=(q,q_1,\ldots,q_l,d)$ is $\varepsilon$, for all $i\in I^c(\sigma)$ satisfying $\mathrm{C}_2$ or $\mathrm{C}_4$, the $i+1$st component of state $\tilde{q}$, i.e., $q_i$, could move by itself with the unobservable event occurrence of $e$.
As shown in (\ref{Eq5-2}), for any $i \in I^c(\sigma)$ satisfying $\mathrm{C}_2$ or $\mathrm{C}_4$,  $q_i$ is updated to $\delta_H(q_i,e)$.
Meanwhile, for any $i\in I^c(\sigma)$ satisfying $\mathrm{C}_5$,  $q_i$ cannot be changed even if the occurrence of $e$ is unobservable to supervisor $i$.}

\textcolor{blue}{Intuitively, for any $i \in I^c(\sigma)$ satisfying $\mathrm{C}_1 \sim \mathrm{C}_4$, $q_i$ tracks a string having the same natural projection as the system string tracked by $q$.
For any $i \in I^c(\sigma)$ satisfying $\mathrm{C}_5$, $q_i$ tracks a string $s_i \in \mathcal{L}(H)$ such that $P_i(s_i) \in \Theta_i^{N_{o,i}}(s)\wedge s_i\sigma \in \mathcal{L}(H_{N_{o,i}}^{aug})$,  where $s$ is the system string tracked by $q$.}

\begin{remark}
\textcolor{blue}{The verifier $V_{\sigma}^{k}$ can be used to verify if (\ref{Eq3}) is true for each $s \in \mathcal{L}^k(H)$ in the following sense: for each  $s \in \mathcal{L}^k(H)$, if there exists a $(s_1,\ldots,s_l)\in \mathcal{T}^{\sigma}_{conf}(s)$ causing a violation of (\ref{Eq3}), one can check that there must exist a $(s'_1,\ldots,s'_l)\in \mathcal{T}^{\sigma}_{conf}(s)$ also causing a violation of (\ref{Eq3}), where $s_i'$ is the shortest prefix of $s_i$ satisfying $P_i(s_i')\in \Theta_i^{N_{o,i}}(s)$ and $s'_i\sigma \in \mathcal{L}(H_{N_{o,i}}^{aug})$.
The verifier $V_{\sigma}^k$ indeed tracks all the $(s,s'_1,\ldots,s'_l)$ regardless of $(s,s_1,\ldots,s_l)$.}
\end{remark}

\begin{remark}
\textcolor{blue}{The verification algorithm proposed in
this paper differs crucially from that presented in \cite{xu20tcns} in the following sense: 
for each $s \in \mathcal{L}^k(H)$, the verifier $V_{\sigma}^{k}$ intends to search string collections $(s_1,\ldots, s_l) \in \mathcal{T}_{conf}^{\sigma}(s)$ that lead to a violation of (\ref{Eq3}), whereas the verifiers proposed in \cite{xu20tcns} search all of the string collections $(s_1,\ldots, s_n,s) \in \underbrace{\mathcal{L}(H)\times \cdots \times  \mathcal{L}(H)}_{n+1}$ such that $s$ may look the same as $s_i$ under observation delays, i.e., $\Theta_i^{N_{o,i}}(s)\cap \Theta_i^{N_{o,i}}(s_i)\neq \emptyset$, $i\in I$.
Meanwhile, benefiting from the partition on $\mathcal{L}(H)$,  our verifiers do not need to record the information of events occurring in history.
Thus, compared with the work of \cite{xu20tcns}, the number of states and transitions of the verifiers proposed in this paper is smaller, and the complexity  of the algorithms proposed in this paper is lower.}
\end{remark}

We use an example to illustrate the construction of $V_{\sigma}^k$.

\begin{figure}
\centering \subfigure[Automaton
$G$]{\label{Fig91}\includegraphics[width=2.7cm]{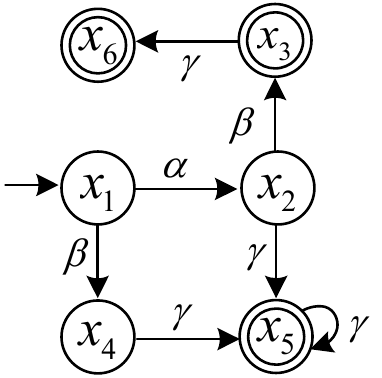}}
\subfigure[Automaton
$H$]{\label{Fig92}\includegraphics[width=2.3cm]{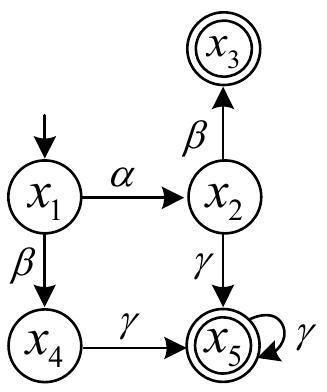}}
\subfigure[$H_{N_{o,1}}^{aug}$ and $H_{N_{o,2}}^{aug}$]{\label{Fig93}\includegraphics[width=2.7cm]{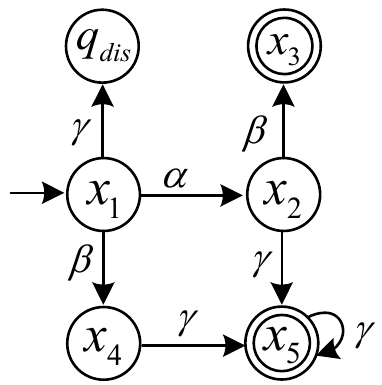}}
\caption{Automata $G$, $H$, and $H_{N_{o,i}}^{aug}$, $i=1,2$.} \label{Fig9}
\end{figure}

\begin{figure}
\centering \subfigure[$V_{\gamma}^{0}$]{\label{Fig10-1}\includegraphics[width=1.6cm]{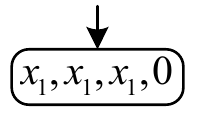}}
\subfigure[$V_{\gamma}^{1}$]{\label{Fig10-2}\includegraphics[width=6.8cm]{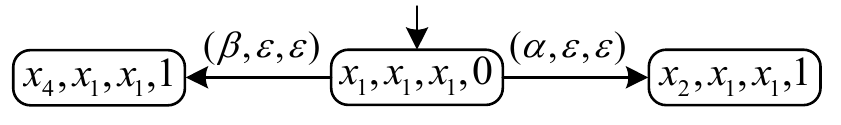}}
\caption{Verifiers $V_{\gamma}^{0}$ and $V_{\gamma}^{1}$  in Example \ref{Ex2}.} \label{Fig4}
\end{figure}

\begin{example}\label{Ex2}

Consider system $G$ and the desired system $H$ that are depicted in Fig.\ref{Fig91} and Fig.\ref{Fig92},
respectively. 
Suppose there are two local supervisors, i.e., $I=\{1,2\}$. 
Let $\Sigma_{o,1}=\{\alpha,\gamma\}$, $\Sigma_{o,2}=\{\beta,\gamma\}$, and $\Sigma_{c,1}=\Sigma_{c,2}=\{\gamma\}$.
Additionally, let $N_{o,1}=2$ and $N_{o,2}=1$.
The augmented automata $H^{aug}_{N_{o,1}}$ and $H^{aug}_{N_{o,2}}$ have the same structure, and are depicted in  Fig.\ref{Fig93}.
By definition, $N=\max\{N_{o,1},N_{o,2}\}=2$.
We now construct $V^0_{\gamma}$ and $V^1_{\gamma}$.

The initial state of $V^0_{\gamma}$ is $\tilde{q}_0=(x_1,x_1,x_1,d=0)$.
Since $d+1>0$, by (\ref{Eq5-1}), no first type of event is defined at $\tilde{q}_0$.
Meanwhile, since $0-d = 0 \le N_{o,1}, N_{o,2}$ and $\gamma$ are active at $x_1$ in both $H^{aug}_{N_{o,1}}$ and $H^{aug}_{N_{o,2}}$, i.e., $\gamma \in \Gamma_{H,N_{o,1}}^{aug}(x_1) \wedge \gamma \in \Gamma_{H,N_{o,2}}^{aug}(x_2)$,  $\mathrm{C}_5$ is satisfied in $\tilde{q}_0$ for $i=1,2$.
By (\ref{Eq5-2}), a second type of event is also not defined at $\tilde{q}_0$.
Therefore, $V^0_{\gamma}$ consists of one single state $\tilde{q}_0$, as depicted in  Fig.\ref{Fig10-1}.
 
The initial state of $V^1_{\gamma}$ also is $\tilde{q}_0=(x_1,x_1,x_1,d=0)$.
Since $1-d=1 \le N_{o,1}, N_{o,2}$, $\gamma \in \Gamma_{H,N_{o,1}}^{aug}(x_1)$, and $\gamma \in \Gamma_{H,N_{o,2}}^{aug}(x_1)$, both supervisors 1 and  2 satisfy $\mathrm{C}_5$ in $\tilde{q}_0$.
Moreover, since $d+1=0 \le 1$ and $\alpha,\beta \in \Gamma_H(x_1)$, by (\ref{Eq5-1}),  first type of transitions are defined at $\tilde{q}_0$ as: $f_{\gamma}^1(\tilde{q}_0,(\alpha,\varepsilon,\varepsilon))=(x_2,x_1,x_1,1)$ and $f_{\gamma}^1(\tilde{q}_0,(\beta,\varepsilon,\varepsilon))=(x_4,x_1,x_1,1)$.
On the other hand, since  $\mathrm{C}_5$ is satisfied in $\tilde{q}_0$ for $i=1,2$, by (\ref{Eq5-2}),  no second type of transition is defined at $\tilde{q}_0$.
Overall, $V^1_{\gamma}$ is constructed in Fig.\ref{Fig10-2}.

\end{example}

Proposition \ref{Prop4} reveals that, to verify delay coobservability, it suffices to check whether or
not the verifier $V_{\sigma}^k$ contains a ``bad'' state  $(q,q_1,\ldots, q_l,k)\in X_{\sigma}^k$ such that $\sigma$ is active at $q$ in $G$ and $q_i$ in $H_{N_{o,i}}^{aug}$ for all $i\in I^c(\sigma)$; however, it is not active at $q$ in $H$.

\begin{proposition}\label{Prop4}
For any $s \in \mathcal{L}^k(H)$, $k=0,1,\ldots,N-1$,  (\ref{Eq3}) is true, i.e., 
\begin{align*}
s\sigma \in \mathcal{L}(G)\setminus \mathcal{L}(H)
\Rightarrow &[(\forall (s_1,\ldots, s_l)\in \mathcal{T}_{conf}^{\sigma}(s))  \notag \\
&(\exists i \in I^{c}(\sigma))s_i\sigma \notin \mathcal{L}(H_{N_{o,i}}^{aug})],
\end{align*}
iff $\not\exists (q,q_1,\ldots,q_l,k)\in X^k_{\sigma}$ such that $\sigma \in \Gamma(q)\setminus\Gamma_H(q)$ and $\sigma \in \Gamma_{H,N_{o,i}}^{aug}(q_i)$ for all $i\in I^c(\sigma)$.
\end{proposition}

\begin{proof}
($\Rightarrow$)
The proof is by contradiction.
Suppose that there exists $(q,q_1,\ldots,q_l,k)\in X^k_{\sigma}$ such that $\sigma \in \Gamma(q)\setminus\Gamma_H(q)$ and $\sigma \in \Gamma_{H,N_{o,i}}^{aug}(q_i)$ for all $i \in I^c(\sigma)$.
Since $(q,q_1,\ldots,q_l,k)\in X^k_{\sigma}$, there exists a $\mu \in \mathcal{L}(V_{\sigma}^k)$ such that $f_{\sigma}^k(x_{\sigma,0}^k,\mu)=(q,q_1,\ldots,q_l,k)$.
We write $f_{\sigma}^k(x_{\sigma,0}^k,\mu^j)=\tilde{x}^{j}=(q^j,q_1^j,\ldots,q_l ^j,n^j)$ for $j=0,1,\ldots, h$, where $h=|\mu|$.
Then,  $q^h=q$, $n^h=k$, and $q_i^h=q_i$ for all $i \in I^c(\sigma)$.
By $k-n^h=0 \le N_{o,i}$ and $\sigma \in \Gamma_{H,N_{o,i}}^{aug}(q_i^h)$,  supervisor $i$ satisfies $\mathrm{C}_5$ in $(q^h,q_1^h,\ldots,q_l^h,n^h)$ for all $i\in I^c(\sigma)$.
w.l.o.g., let $\mu^{k_i}$ be the shortest prefix of $\mu$ such that the supervisor $i$ satisfies $\textup{C}_5$ in state $\tilde{x}^{k_i}=(q^{k_i},q_1^{k_i},\ldots,q_l^{k_i},n^{k_i})$, i.e., $\sigma \in \Gamma_{H,N_{o,i}}^{aug}(q_i^{k_i})$ and $k-n^{k_i} \le N_{o,i}$.
By the definition of $V_{\sigma}^k$, we know that $n^{k_i}$ records the length of the system string tracked by $q^{k_i}$, and $q_i^{k_i}$ tracks a string in $\mathcal{L}(H)$ having the same natural projection as the system string tracked by $q^{k_i}$.
Assuming that $q^h$ tracks $s \in \mathcal{L}^k(H)$, i.e., $q^{h}=\delta_H(q_0,s)$, then $q^{k_i}=\delta_H(q_0,s^{n^{k_i}})$.
Also, assuming that $q_i^{k_i}$ tracks $s_i\in \mathcal{L}(H)$, i.e., $q_i^{k_i}=\delta_H(q_0,s_i)$, then $P_i(s^{n^{k_i}})=P_i(s_i)$.
Since $k-n^{k_i} \le N_{o,i}$ and $|s|=k$, $P_i(s^{n^{k_i}})=P_i(s_i) \in \Theta_i^{N_{o,i}}(s)$ for $i\in I^c(\sigma)$.
Therefore, $(s_1,\ldots,s_l)\in \mathcal{T}_{conf}^{\sigma}(s)$.
Moreover, since $s\sigma \in \mathcal{L}(G)\setminus \mathcal{L}(H)$ and  $s_i\sigma \in \mathcal{L}(H_{N_{o,i}}^{aug})$ for all $i\in I^c(\sigma)$,  it contradicts (\ref{Eq3}).

($\Leftarrow$) The proof is also by contradiction. 
Suppose that (\ref{Eq3}) is not true. 
That is, there exist $s\in \mathcal{L}^k(H)$ and  $(s_1,\ldots,s_l) \in \mathcal{T}_{conf}^{\sigma}(s)$ such that $s\sigma \in \mathcal{L}(G)\setminus\mathcal{L}(H)$ and $s_i\sigma \in \mathcal{L}(H_{N_{o,i}}^{aug})$ for all $i\in I^c(\sigma)$.
w.l.o.g., let $s_i$ be the shortest string in $\mathcal{L}(H)$ satisfying  $P_i(s_i) \in \Theta_i^{N_{o,i}}(s)$ and $s_i\sigma \in \mathcal{L}(H_{N_{o,i}}^{aug})$ for all $i\in I^c(\sigma)$.
We now prove that there exists $(q,q_1,\ldots, q_l,k) \in X_{\sigma}^k$ with $q=\delta_H(q_0,s)$ and $q_i=\delta_H(q_0,s_i)$ for all $i\in I^c(\sigma)$.
Since $P_i(s_i) \in \Theta_i^{N_{o,i}}(s)$, we know $P_i(s_i)=P_i(s^{d_i})$ for some $d_i\in \{\max\{0,k-N_{o,i}\},\ldots,k\}$.
For brevity, we introduce the following claim.

\begin{claim}\label{Cla1}
$\forall d=0,1,\ldots, k$, there exists a $(q,q_1,\ldots,q_l,d)\in X^k_{\sigma}$ such that (i) $q=\delta_H(q_0,s^d)$ and, (ii) for all $i \in I^c(\sigma)$, if $d \ge d_i$,  $q_i=\delta_H(q_0,s_i)$, and if $d<d_i$, $q_i=\delta_H(q_0,t_i)$, where $t_i$ is the longest prefix of $s_i$ with $P_i(t_i )=P_i(s^d)$.
\end{claim}

The proof of Claim \ref{Cla1} is given in  Appendix A.
Since $d_i \in \{\max\{0,k-N_{o,i}\},\ldots,k\}$, we have $k \ge d_i$ for all $i \in I^c(\sigma)$.
By Claim \ref{Cla1}, there exists $(q,q_1,\ldots,q_l,d=k) \in X^k_{\sigma}$ such that $q=\delta_H(q_0,s^k)=\delta_H(q_0,s)$ and $q_i=\delta_H(q_0,s_i)$ for all $i \in I^c(\sigma)$.
Moreover, since $s\sigma \in \mathcal{L}(G)\setminus\mathcal{L}(H)$ and $s_i\sigma \in \mathcal{L}(H_{N_{o,i}}^{aug})$, we have $\sigma \in \Gamma(q)\setminus\Gamma_H(q)$ and $\sigma \in \Gamma_{H,N_{o,i}}^{aug}(q_i)$  for all $i\in I^c(\sigma)$, which is a contradiction.
\end{proof}

\begin{figure*}[htbp]
	\begin{center}
		\includegraphics[width=18cm]{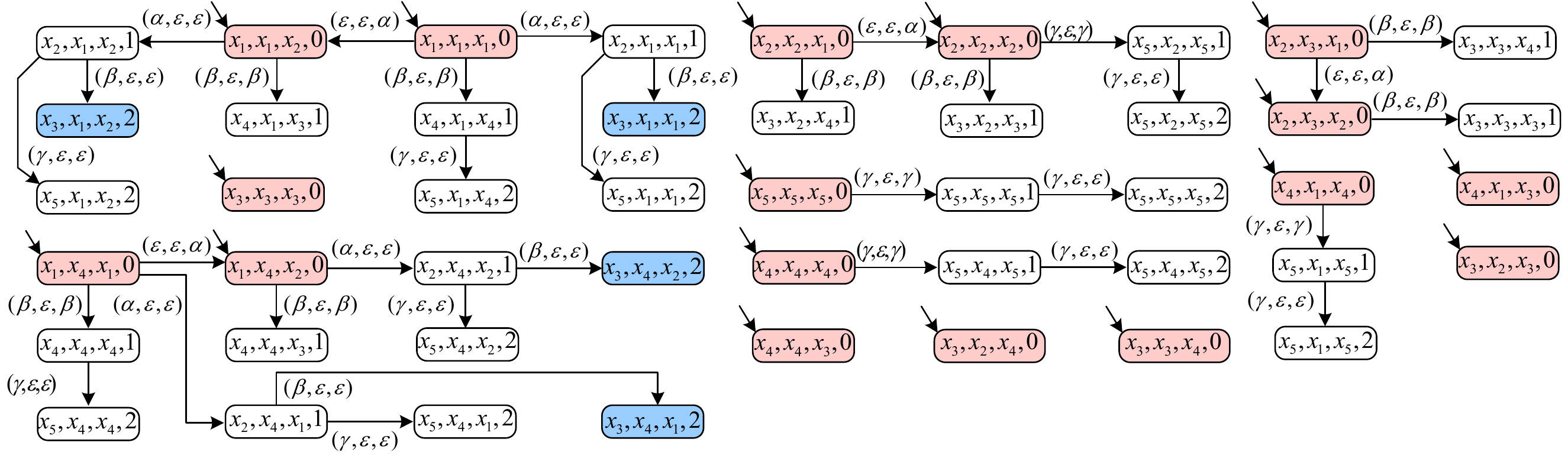}    
		\caption{Verifier $V^2_{\gamma}$ in Example \ref{Ex2}} 
		\label{Fig10-3}
	\end{center}
\end{figure*}

We next consider the construction of $V_{\sigma}^N$.
The construction of $V_{\sigma}^N$ is very similar to that of $V_{\sigma}^k$, $k=0,1,\ldots, N-1$.
The main difference between them is that  the initial state of $V_{\sigma}^k$ is a single state but the initial states of $V_{\sigma}^N$ are a set of states.

For all $s \in \mathcal{L}^N(H)$, to track $t=s_{-N} \in \mathcal{L}(H)$ and $s_i \in \mathcal{L}(H)$ with $P_i(t)=P_i(s_i)$, $i\in I^c(\sigma)$, we define
\begin{align*}
\Xi=&\{(q,q_1,\ldots,q_l): [(\exists t \in \mathcal{L}(H))q=\delta_H(q_0,t)]\wedge \\
&[(\forall i \in I^c(\sigma))(\exists s_i\in \mathcal{L}(H))q_i=\delta_H(q_0,s_i)\wedge P_i(t)=P_i(s_i)]\},
\end{align*}
as the set of confusable state vectors under $P_1,\ldots, P_l$.\footnote{$\Xi$ can be calculated by constructing a $M$-machine that was studied in \cite{rudie95tac,yoo04tac,yoo02deds} for the verification of coobservability.}
The initial states of $V_{\sigma}^N$ are defined as: $$x_{\sigma,0}^N=\{(q,q_1,\ldots,q_l,0):(q,q_1,\ldots,q_l) \in \Xi\}.$$
By looking $N$ steps forward from all the initial states, we construct $V_{\sigma}^N=(X_{\sigma}^N,\tilde{\Sigma},f_{\sigma}^N, x_{\sigma,0}^N)$, where

\begin{itemize}
\item{
the state space is $X_{\sigma}^N \subseteq \underbrace{Q_{H} \times  \cdots \times Q_{H}}_{l+1} \times [0,N]$;}
\item{
$\tilde{\Sigma} \subseteq \underbrace{(\Sigma \cup \{\varepsilon\}) \times  \cdots \times (\Sigma \cup \{\varepsilon\})}_{l+1} \setminus  \{(\underbrace{\varepsilon, \ldots, \varepsilon}_{l+1})\}$ is the event set;}
\item{
$x_{\sigma,0}^N=\{(q,q_1,\ldots,q_l,0):(q,q_1,\ldots,q_l)\in \Xi\}$ is a set of initial states;}
\item{
the transition function  $f_{\sigma}^N: X_{\sigma}^N \times
\tilde{\Sigma} \rightarrow X_{\sigma}^N$ can be specified in the same way as we specify $f_{\sigma}^k$.
Specifically, for each state $\tilde{q}=(q,q_1,\ldots,q_l,d)\in X_{\sigma}^N$ and each event $e \in \Sigma$, we need to consider the following five cases for each $i\in I^c(\sigma)$: 

{$\mathrm{D}_1$: $N-d >N_{o,i}$ and $e \in \Sigma_{o,i}$;}

{$\mathrm{D}_2$: $N-d >N_{o,i}$ and $e \in  \Sigma_{uo,i} $;} 

{$\mathrm{D}_3$: $N-d \le N_{o,i}$ and $\sigma \notin \Gamma_{H,N_{o,i}}^{aug}(q_i)$ and $e \in \Sigma_{o,i}$;}

{$\mathrm{D}_4$: $N-d \le N_{o,i}$ and $\sigma \notin \Gamma_{H,N_{o,i}}^{aug}(q_i)$ and $e \in  \Sigma_{uo,i}$;}

{$\mathrm{D}_5$: $N-d \le N_{o,i}$ and $\sigma \in \Gamma_{H,N_{o,i}}^{aug}(q_i)$.}

Then, the following two types of transitions are defined in $V_{\sigma}^N$:

\begin{enumerate}
\item if $d+1\le N$, $\delta_H(q,e)$ is defined, and for each $i\in I^c(\sigma)$, $\delta_H(q_i,e)$ is defined for $\mathrm{D}_1$ or $\mathrm{D}_3$. Then, a first type of transition is defined at $\tilde{q}$ as:
\begin{align}\label{Eq7}
&f_{\sigma}^N((q,q_1,\ldots,q_l,d),(e,e_1,\ldots,e_l))=\notag\\
&(\delta_H(q,e),\delta_H(q_1,e_1),\ldots,\delta_H(q_l,e_l),d+1),
\end{align}
where, for all $i \in I^c(\sigma)$
\lfteqn
e_i= \begin{cases}
e & \text{if}\ \mathrm{D}_1\ \text{or}\ \mathrm{D}_3 \\
\varepsilon & \text{if}\ \mathrm{D}_2 \ \text{or}\ \mathrm{D}_4 \ \text{or}\ \mathrm{D}_5.
\end{cases}
 \ndeqn

\item for each $i\in I^c(\sigma)$, if  $\delta_H(q_i,e)$ is defined for $\mathrm{D}_2$ or $\mathrm{D}_4$, a second type of transition is defined at $\tilde{q}$ as:
\begin{align}\label{Eq8}
&f_{\sigma}^N((q,q_1,\ldots,q_l,d),(\varepsilon,\varepsilon,\ldots, \varepsilon, \underset{(i+1)^{st}}{e}, \varepsilon, \ldots,\varepsilon))=\notag\\
&(q, q_1,\ldots,q_{i-1}, \delta_H(q_{i},e), q_{i+1}, \ldots, q_l,d).
\end{align}
\end{enumerate}
}\end{itemize}

We illustrate the construction of $V_{\sigma}^N$ using the following example.

\begin{example}\label{Ex3}

Continuing with Example \ref{Ex2}, since $N_{o,1}=2$ and $N_{o,2}=1$, we have $N=\max{\{N_{o,1},N_{o,2}\}}=2$.
The verifier $V_{\gamma}^2$ is constructed in Fig.\ref{Fig10-3}.
All the initial states of $V_{\gamma}^2$ are highlighted in red.
Let us take the initial state $\tilde{q}_0=(x_2,x_3,x_1,d=0)$ in $V_{\gamma}^2$ as an example.

For state $\tilde{q}_0=(x_2,x_3,x_1,d=0)$ and event $\beta \in \Sigma$, since $N-d=2 \le N_{o,1}$, $\gamma \notin \Gamma_{H,N_{o,1}}^{aug}(x_3)$, and $\beta \in \Sigma_{uo,1}$,  supervisor 1 satisfies $\mathrm{D}_4$ in $\tilde{q}_0$.
Meanwhile, since $N-d=2 > N_{o,2}$ and $\beta \in \Sigma_{o,1}$,  supervisor 2 satisfies $\mathrm{D}_1$ in $\tilde{q}_0$.
Therefore,  by (\ref{Eq7}), a first type of transition is defined at $\tilde{q}_0$ as $f_{\gamma}^2(\tilde{q}_0,(\beta,\varepsilon,\beta))=(x_3,x_3,x_4,1)$.

For state $\tilde{q}_0=(x_2,x_3,x_1,d=0)$ and event $\alpha \in \Sigma$, since $N-d=2 > N_{o,2}$, $\gamma \in \Gamma_{H,N_{o,2}}^{aug}(x_1)$, and $\alpha \in \Sigma_{uo,2}$, supervisor 2 satisfies $\mathrm{D}_4$ in $\tilde{q}_0$.
Since $\delta_H(x_1,\alpha)=2$, by (\ref{Eq8}), a second type of transition is defined at $\tilde{q}_0$ as: $f_{\gamma}^2(\tilde{q}_0,(\varepsilon,\varepsilon,\alpha))=(x_2,x_3,x_2,0)$.

In this way, we can define all the transitions of $V_{\gamma}^2$.

\end{example}

\begin{proposition}\label{Prop3}
For any $s \in \mathcal{L}^N(H)$, (\ref{Eq3}) is true, i.e., 
\begin{align*}
s\sigma \in \mathcal{L}(G)\setminus \mathcal{L}(H)
\Rightarrow &[(\forall (s_1,\ldots, s_l)\in \mathcal{T}^{\sigma}_{conf}(s))  \notag \\
&(\exists i \in I^{c}(\sigma))s_i\sigma \notin \mathcal{L}(H_{N_{o,i}}^{aug})],
\end{align*}
iff $\not\exists (q,q_1,\ldots,q_l,N)\in X^N_{\sigma}$ such that $\sigma \in \Gamma(q)\setminus\Gamma_H(q)$ and $\sigma \in \Gamma_{H,N_{o,i}}^{aug}(q_i)$ for all $i\in I^c(\sigma)$.
\end{proposition} 
\begin{proof}
($\Rightarrow$)
The proof is by contradiction.
Suppose that there exists $(q,q_1,\ldots,q_l,N)\in X^N_{\sigma}$ such that $\sigma \in \Gamma(q)\setminus\Gamma_H(q)$ and $\sigma \in \Gamma_{H,N_{o,i}}^{aug}(q_i)$ for all $i\in I^c(\sigma)$.
Since $(q,q_1,\ldots,q_l,N)\in X^N_{\sigma}$, there exist $\tilde{x}_0=(x,x_1,\ldots,x_l,0) \in x_{\sigma,0}^N$ and $\mu \in \mathcal{L}(V_{\sigma}^N)$ such that $f_{\sigma}^N(\tilde{x}_0,\mu)=(q,q_1,\ldots,q_l,N)$.
Since $(x,x_1,\ldots,x_l,0) \in x_{\sigma,0}^N$, $(x,x_1,\ldots,x_l)\in \Xi$.
Thus, there exist $s,s_1,\ldots,s_l \in \mathcal{L}(H)$ such that $\delta_H(q_0,s)=x$ and $\delta_H(q_0,s_i)=x_i$ and $P_i(s)=P_i(s_i)$ for all $i\in I^c(\sigma)$.
We write $f_{\sigma}^N(\tilde{x}_0,\mu^j)=\tilde{x}^{j}=(q^j,q_1^j,\ldots,q_l^j,n^j)$ for $j=0,1,\ldots, h$, where $h=|\mu|$.
Then, $q^h=q$, $n^h=N$, and $q_i^h=q_i$ for all $i\in I^c(\sigma)$.
Since $N-n^h \le N_{o,i}$ and $\sigma \in \Gamma_{H,N_{o,i}}^{aug}(q_i^h)$, we know supervisor $i$ satisfies $\mathrm{D}_5$ in $(q^h,q_1^h,\ldots,q_n^h,n^h)$ for all $i\in I^c(\sigma)$.
w.l.o.g.,  let $\mu^{k_i}$ be the shortest prefix of $\mu$ such that the supervisor $i$ satisfies $\textup{D}_5$ in state $\tilde{x}^{k_i}=f_{\sigma}^N(\tilde{x}_0,\mu^{k_i})=(q^{k_i},q_1^{k_i},\ldots,q_n^{k_i},n^{k_i})$, i.e., $\sigma \in \Gamma_{H,N_{o,i}}^{aug}(q_i^{k_i})$ and $N-n^{k_i} \le N_{o,i}$.
Assume that $q^h$ tracks $st \in \mathcal{L}^N(H)$ with $|t|=N$ such that $x=\delta_H(q_0,s)$ and $q^h=\delta_H(x,t)$.
Then, we have $q^{k_i}=\delta_H(x,t^{n^{k_i}})=\delta_H(q_0,st^{n^{k_i}})$.
Also, assume that $q_i^{k_i}$ tracks some $s_it_i\in \mathcal{L}(H)$ such that $x_i=\delta_H(q_0,s_i)$ and $q_i^{k_i}=\delta_H(x_i,t_i)$.
By the definition of $V_{\sigma}^N$,   $P_i(t^{n^{k_i}})=P_i(t_i)$.
Moreover, since $P_i(s)=P_i(s_i)$, $P_i(st^{n^{k_i}})=P_i(s_it_i)$.
Since $N-n^{k_i} \le N_{o,i}$ and $|t|=N$, $P_i(st^{n^{k_i}})=P_i(s_it_i) \in \Theta_i^{N_{o,i}}(st)$.
Therefore, $(s_1t_1,\ldots,s_lt_l)\in \mathcal{T}_{conf}^{\sigma}(st)$.
Moreover, since $st\sigma \in \mathcal{L}(G)\setminus \mathcal{L}(H)$ and  $s_it_i\sigma \in \mathcal{L}(H_{N_{o,i}}^{aug})$ for all $i\in I^c(\sigma)$, it contradicts (\ref{Eq3}).

($\Leftarrow$) The proof is also by contradiction. 
Suppose that there exist $s \in \mathcal{L}^N(H)$ and $(s_1,\ldots, s_l) \in \mathcal{T}^{\sigma}_{conf}(s)$ such that $s\sigma \in \mathcal{L}(G)\setminus\mathcal{L}(H)$ and $s_i\sigma \in \mathcal{L}(H_{N_{o,i}}^{aug})$ for all $i\in I^c(\sigma)$.
w.l.o.g., let  $s_i$ be the shortest string in $\mathcal{L}(H)$ such that $P_i(s_i) \in \Theta_i^{N_{o,i}}(s)$ and $s_i\sigma \in \mathcal{L}(H_{N_{o,i}}^{aug})$ for all $i\in I^c(\sigma)$.
We next prove there exists $(q,q_1,\ldots, q_l,N) \in X_{\sigma}^N$ such that $q=\delta_H(q_0,s)$ and $q_i=\delta_H(q_0,s_i)$ for all $i\in I^c(\sigma)$.

Since $|s| \ge N$, w.l.o.g., we write $s=s_{-N}w$ with $w=e_1 \cdots e_N$.
Since $P_i(s_i) \in \Theta_i^{N_{o,i}}(s)$, we know $P_i(s_i)=P_i(s_{-N}w^{d_i})$ for some $d_i\in \{N-N_{o,i},\ldots,N\}$.
We write $s_i=u_iw_i$ such that $P_i(u_i)=P_i(s_{-N})$ and $P_i(w_i)=P_i(w^{d_i})$.
We also write $x=\delta_H(q_0,s_{-N})$, $x_i=\delta_H(q_0,u_i)$ for all $i\in I^c(\sigma)$.
Since $P_i(s_{-N})=P_i(u_i)$ for all $i\in I^c(\sigma)$, $(x,x_1,\ldots,x_l)\in \Xi$.
Hence, $(x,x_1,\ldots,x_l,0) \in x^N_{\sigma,0}$.
For brevity, we introduce the following claim.

\begin{claim}\label{Cla2}
$\forall d=0,1,\ldots,N$, there exists $(q,q_1,\ldots,q_l,d)\in X^N_{\sigma}$ such that (i) $q=\delta_H(q_0,s_{-N}w^d)$ and, (ii) for all $i \in I^c(\sigma)$, if $d \ge d_i$,  $q_i=\delta_H(q_0,s_i)$, and if $d<d_i$, $q_i=\delta_H(q_0,u_iv_i)$, where $u_iv_i$ is the longest prefix of $u_iw_i=s_i$ with $P_i(u_iv_i )=P_i(s_{-N}w^d)$.
\end{claim}

Claim \ref{Cla2} is proven in Appendix B.
Since $N \ge d_i$, by Claim \ref{Cla2}, there exists $(q,q_1,\ldots,q_l,N) \in X^N_{\sigma}$ such that $q=\delta_H(q_0,s_{-N}w)$ and $q_i=\delta_H(q_0,s_i)$ for all $i \in I^c(\sigma)$.
Moreover, since $s_{-N}w\sigma \in \mathcal{L}(G)\setminus\mathcal{L}(H)$ and $s_i\sigma \in \mathcal{L}(H_{N_{o,i}}^{aug})$, we have $\sigma \in \Gamma(q)\setminus\Gamma_H(q)$ and $\sigma \in \Gamma_{H,N_{o,i}}^{aug}(q_i)$  for all $i\in I^c(\sigma)$, which is a contradiction.
\end{proof}

The following theorem provides the means to verify the delay coobservability of $\mathcal{L}(H)$ using $V_{\sigma}^k$, $k=0,1,\ldots,N$.

\begin{theorem}\label{Theo2}
$\mathcal{L}(H)$ is delay coobservable w.r.t. $N_{o,1}, \ldots, N_{o,n}$, $\sigma \in \Sigma_c$, and $\mathcal{L}(G)$ iff, 
for any $k=0,1,\ldots,N$, there does not exist a state $(q,q_1,\ldots,q_n,k)\in X^k_{\sigma}$ such that $\sigma \in \Gamma(q)\setminus\Gamma_H(q)$ and $\sigma \in \Gamma_{H,N_{o,i}}^{aug}(q_i)$ for all $i\in I^c(\sigma)$. 

\end{theorem}
\begin{proof}
The proof directly follows from Theorem \ref{Theo1} and Propositions \ref{Prop4} and \ref{Prop3}.
\end{proof}

Next, we use an example to illustrate how to verify delay coobservability using the proposed approach.

\begin{example}\label{Ex4}
\textcolor{blue}{We continue with Examples \ref{Ex2} and \ref{Ex3}. 
By assumption, we know $\Sigma_{o,1}=\{\alpha,\gamma\}$, $\Sigma_{o,2}=\{\beta,\gamma\}$, $N_{o,1}=2$, and $N_{o,2}=1$.
We construct verifiers $V_{\gamma}^0$, $V_{\gamma}^1$, and $V_{\gamma}^2$ in Fig.\ref{Fig10-1}, Fig.\ref{Fig10-2}, and Fig.\ref{Fig10-3}, respectively.
To verify delay coobservability,  it is necessary to check if there exists a ``bad'' state $(q,q_1,q_2,k)$ in $V_{\gamma}^0$, $V_{\gamma}^1$, and $V_{\gamma}^2$ such that $\gamma \in \Gamma(q)\setminus \Gamma_H(q)$, $\gamma \in \Gamma_{H,N_{o,1}}^{aug}(q_1)$, and $ \gamma \in \Gamma_{H,N_{o,2}}^{aug}(q_2)$.
The verifier $V_{\gamma}^2$ has several ``bad'' states, for instance, $(x_3,x_1,x_1,2)$, $(x_3,x_1,x_2,2)$, $(x_3,x_4,x_1,2)$, and $(x_3,x_4,x_2,2)$ (highlighted in blue in Fig.\ref{Fig10-3}).
By Theorem \ref{Theo2}, the existence of ``bad'' states implies that $\mathcal{L}(H)$ is not delay coobservable w.r.t. ${N_{o,1}}$, ${N_{o,2}}$, and $\mathcal{L}(G)$.}

\textcolor{blue}{As pointed out in Remark 1, a ``bad'' state $(q,q_1,q_2,k)$ tracks some $s \in \mathcal{L}^k(H)$ and $(s_1,s_2)\in \mathcal{T}_{conf}^{\sigma}(s)$  violating (\ref{Eq3}).
It is shown in Fig.\ref{Fig10-3} that the ``bad'' state $(x_3,x_1,x_1,2)$ of $V_{\gamma}^2$ tracks $s=\alpha\beta$ and $(s_1,s_2)=(\varepsilon,\varepsilon)\in \mathcal{T}_{conf}^{\sigma}(s)$.
By the definition of $\mathcal{T}^{\sigma}_{conf}(\cdot)$, when $s=\alpha\beta$ occurs in $G$, the supervisor $i$ may see $P_i(s_i)=P_i(\varepsilon) \in \Theta_i^{N_{o,i}}(s)$, $i=1,2$.
Since $s_i\gamma \in \mathcal{L}(H_{N_{o,i}}^{aug})$, there exists $t_i=\alpha \in \Sigma^{\le N_{o,i}}$ such that  $s_it_i\gamma \in \mathcal{L}(H)$ for $i=1,2$.
Therefore,  when $s=\alpha\beta$ occurs in $G$, control conflicts may arise for both supervisors $1$ and $2$, because $s\gamma \in \mathcal{L}(G)\setminus \mathcal{L}(H)$, $s_it_i\gamma \in \mathcal{L}(H)$, and $P_i(s_i) \in \Theta_i^{N_{o,i}}(s) \cap \Theta_i^{N_{o,i}}(s_it_i) \neq \emptyset$ for $i=1,2$.}

\end{example}

\section{\textcolor{blue}{Computational complexity}}

In this section, we analyze the computational complexity of the proposed approach, and compare it with that of the existing one in the literature.

The computational complexity is determined by the number of transitions of all the verifiers with the verification being done for one controllable event at a time,
The verification can be done for one controllable event at a time.
As shown in Section III, for each controllable event $\sigma \in \Sigma_c$, we need to construct verifiers $V_{\sigma}^1,\ldots, V_{\sigma}^N$, where $N=\max\{N_{o,1},\ldots, N_{o,l}\} \le \max\{N_{o,1},\ldots, N_{o,n}\}$.
By the definition of $X_{\sigma}^k$, the number of states of $V_{\sigma}^k$ is upper bounded by $(k+1)\times |Q_H|^{n+1}$.
Since there could be $(n+1)\times |\Sigma|$ transitions in each state of $V_{\sigma}^k$,  the number of transitions in $V_{\sigma}^k$ is upper bounded by  $(n+1)\times(k+1)\times |Q_H|^{n+1} \times |\Sigma|$.
Therefore, the worst-case complexity for constructing $V_{\sigma}^1,\ldots, V_{\sigma}^N$ is $\mathcal{O}((n+1) \times |Q_H|^{n+1}\times |\Sigma|  \times \sum_{k=0}^{\max\{N_{o,1},\ldots, N_{o,n}\}} (k+1))$.
Since each controllable event can be verified independently, the worst-case complexity of verifying delay coobservability is $\mathcal{O}((n+1) \times |Q_H|^{n+1}\times |\Sigma|^2  \times \sum_{k=0}^{\max\{N_{o,1},\ldots, N_{o,n}\}} (k+1))$,  which is polynomial w.r.t. $|Q_H|$, $|\Sigma|$, and ${\max\{N_{o,1},\ldots, N_{o,n}\}}$, and is exponential only w.r.t. $n$.

Next, let us briefly recall the algorithm proposed in \cite{xu20tcns} for the verification of delay coobservability.

For each combination of the fixed observation delays $\omega=(m_1,\ldots ,m_n)\in [0,N_{o,1}]\times \cdots \times [0,N_{o,n}]$, the authors in  \cite{xu20tcns} constructed a verifier $V^{\omega}$ to track all the $(s_1,\ldots, s_n, s)\in \underbrace{\mathcal{L}(H)\times \cdots \times \mathcal{L}(H)}_{n+1}$ such that $(\forall i \in I)P_i(s_i)=P_i(s_{-m_i})$.
The complexity for constructing $V^{\omega}$ is the order of $(n+1)\times |Q_H|^{n+1}\times|\Sigma|^{\max\{N_{o,1},\ldots, N_{o,n}\}+1}$.\footnote{The time complexity for constructing a verifier was originally described as $\mathcal{O}(|Q_H|^{n+1}\times|\Sigma|^{\max\{N_{o,1},\ldots, N_{o,n}\}})$ in \cite{xu20tcns} because it only considers the state space of the verifier.
However, to construct the verifier $V^{\omega}$, we need to consider the transitions in the verifier.
In each state of $V^{\omega}$, there could be $(n+1)\times|\Sigma|$ transitions. 
Therefore,  the original time complexity should be multiplied by $(n+1)\times|\Sigma|$.}
Since $|[0,N_{o,1}]|\times \cdots \times |[0,N_{o,n}]|=\prod_{i=1}^n(N_{o,i}+1)$, the number of verifiers to be constructed in \cite{xu20tcns} is  $\prod_{i=1}^n(N_{o,i}+1)$.
Therefore, the worst-case complexity to verify delay coobservability in \cite{xu20tcns} is $\mathcal{O}((n+1)\times |Q_H|^{n+1}\times|\Sigma|^{\max\{N_{o,1},\ldots, N_{o,n}\}+1}\times\prod_{i=1}^n(N_{o,i}+1))$, which is polynomial w.r.t. $|Q_H|$ and $|\Sigma|$  but exponential w.r.t. $n$ and $\max\{N_{o,1},\ldots, N_{o,n}\}$. 

\begin{table}
\begin{center}
\caption{Number of states and transitions of our verifiers.}
{\begin{tabular}{ccc} \toprule
Verifier & Number of states & Number of transitions \\ \midrule
$V_{\gamma}^0$ & 1 & 0 \\
$V_{\gamma}^1$ & 3 & 2 \\
$V_{\gamma}^2$ & 47 & 34 \\ \midrule
In total & 51 & 36 \\
 \bottomrule
\end{tabular}}
\end{center}
\label{table1}
\end{table}
\begin{table}
\begin{center}
\caption{Number of states and transitions of verifiers proposed in \cite{xu20tcns}.}
{\begin{tabular}{ccc} \toprule
Verifier & Number of states & Number of transitions \\ \midrule
$V^{(0,0)}$ & 17 & 25 \\
$V^{(0,1)}$ & 21 & 34 \\
$V^{(1,0)}$ & 21 & 30 \\
$V^{(1,1)}$ & 31 & 50 \\
$V^{(2,0)}$ & 25 & 38 \\
$V^{(2,1)}$ & 31 & 52 \\ \midrule
In total & 146 & 229 \\
 \bottomrule
\end{tabular}}
\end{center}
\label{table2}
\end{table}

\textcolor{blue}{Consider again Examples \ref{Ex2}, \ref{Ex3}, and \ref{Ex4}.
We now compare the numbers of states and transitions of  verifiers proposed in this paper with that of verifiers proposed in \cite{xu20tcns}.
To verify delay coobservability, we construct verifiers $V_{\gamma}^0$, $V_{\gamma}^1$, and $V_{\gamma}^2$ in Fig.\ref{Fig10-1}, Fig.\ref{Fig10-2}, and Fig.\ref{Fig10-3}, respectively.
Table I summarizes the numbers of states and transitions of verifiers $V_{\gamma}^0$, $V_{\gamma}^1$, and $V_{\gamma}^2$.
There are 51 states and 36 transitions in total.}

\textcolor{blue}{However, since $N_{o,1}=2$ and $N_{o,2}=1$, the approach proposed in \cite{xu20tcns} constructs $V^{(0,0)}$, $V^{(0,1)}$, $V^{(1,0)}$, $V^{(1,1)}$, $V^{(2,0)}$, and $V^{(2,1)}$. 
Due to space limitations, we will not provide the construction details of these verifiers.
The reader is referred to \cite{xu20tcns} for more information.
Table II gives the numbers of states and transitions of these verifiers.
There are 146 states and 229 transitions in total.
Clearly, our approach is more efficient as it has lower computational complexity.}

\section{\textcolor{blue}{Application in urban traffic control}} \label{Exa}

In this section, we use a practical example to show an application of results derived in this paper.

\subsection{Research object}

\begin{figure}
	\begin{center}
		\includegraphics[width=8.8cm]{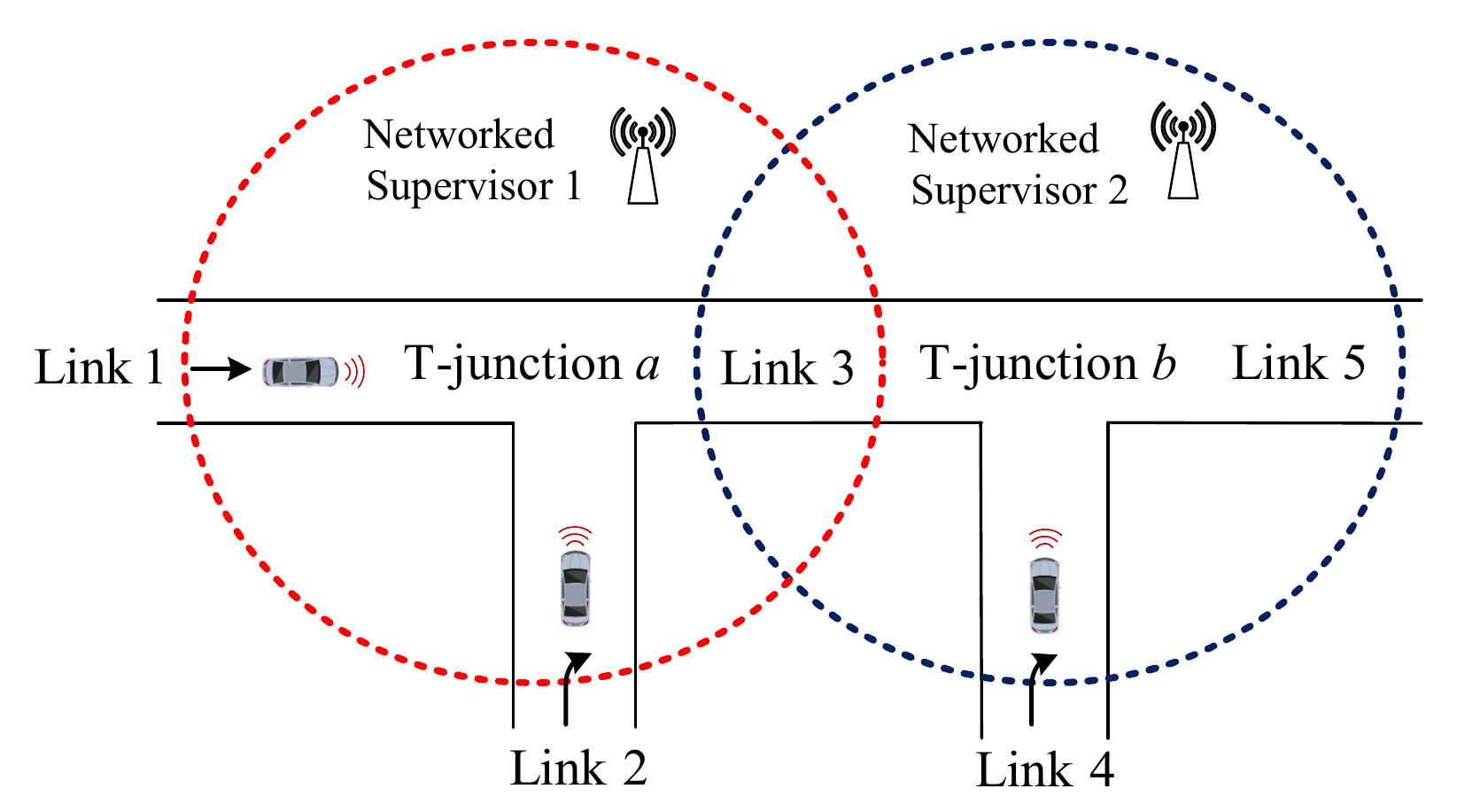}    
		\caption{A small urban traffic network.} 
		\label{Fig4}
	\end{center}
\end{figure}

As shown in Fig.\ref{Fig4}, we consider a small urban traffic network consisting of two unsignalized
T-junctions (T-junctions $a$ and $b$), and five one-lane links (Links $1,\ldots, 5$).
Note that all the links are not in scale and can be longer than represented.
We assume, in this example, (i) a vehicle coming from Link $1$ (Link $2$)  needs to pass through T-junctions $a$ followed by  $b$  to complete a straight movement (a right turn); (ii) a vehicle coming from Link $4$ needs to pass through T-junction $b$ to complete a right turn.
Networked Supervisors 1 and 2 are used to collect signals sent by the vehicles and control their movements.
They are distributed in T-junctions $a$ and $b$, respectively.
To simplify the problem, we assume that Link 3 can accommodate one vehicle at most, i.e., the maximum vehicle queue length for Link 3 is 1.
\subsection{System model}

As shown in Fig.\ref{Fig6}, ${G}_{a}$ and ${G}_{b}$ are system models for T-junctions $a$ and $b$, respectively.
Events in ${G}_a$ and $G_b$ are  defined as follows: $\alpha_1$: a vehicle coming from Link 1 approaches T-junction $a$; $\alpha_2$: a vehicle coming from Link 2 approaches T-junction $a$; $\beta_1$: a vehicle coming from Link 1 passes through T-junction $a$, drives into Link 3, and approaches T-junction $b$; $\beta_2$: a vehicle coming from Link 2 passes through T-junction $a$, drives into Link 3, and approaches T-junction $b$; $\beta_3$: a vehicle coming from Link 4 approaches T-junction $b$; $\gamma_1$: a vehicle coming from Links 1 or 2 passes through T-junction $b$; $\gamma_2$: a vehicle coming from Link 4 passes through T-junction $b$.

\begin{figure}
\centering \subfigure[Automaton
$G_a$]{\label{Fig51}\includegraphics[width=3.6cm]{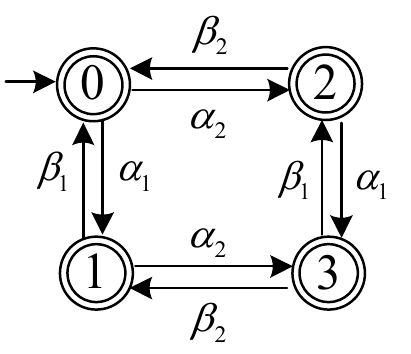}}
\subfigure[Automaton
$G_b$]{\label{Fig52}\includegraphics[width=4.1cm]{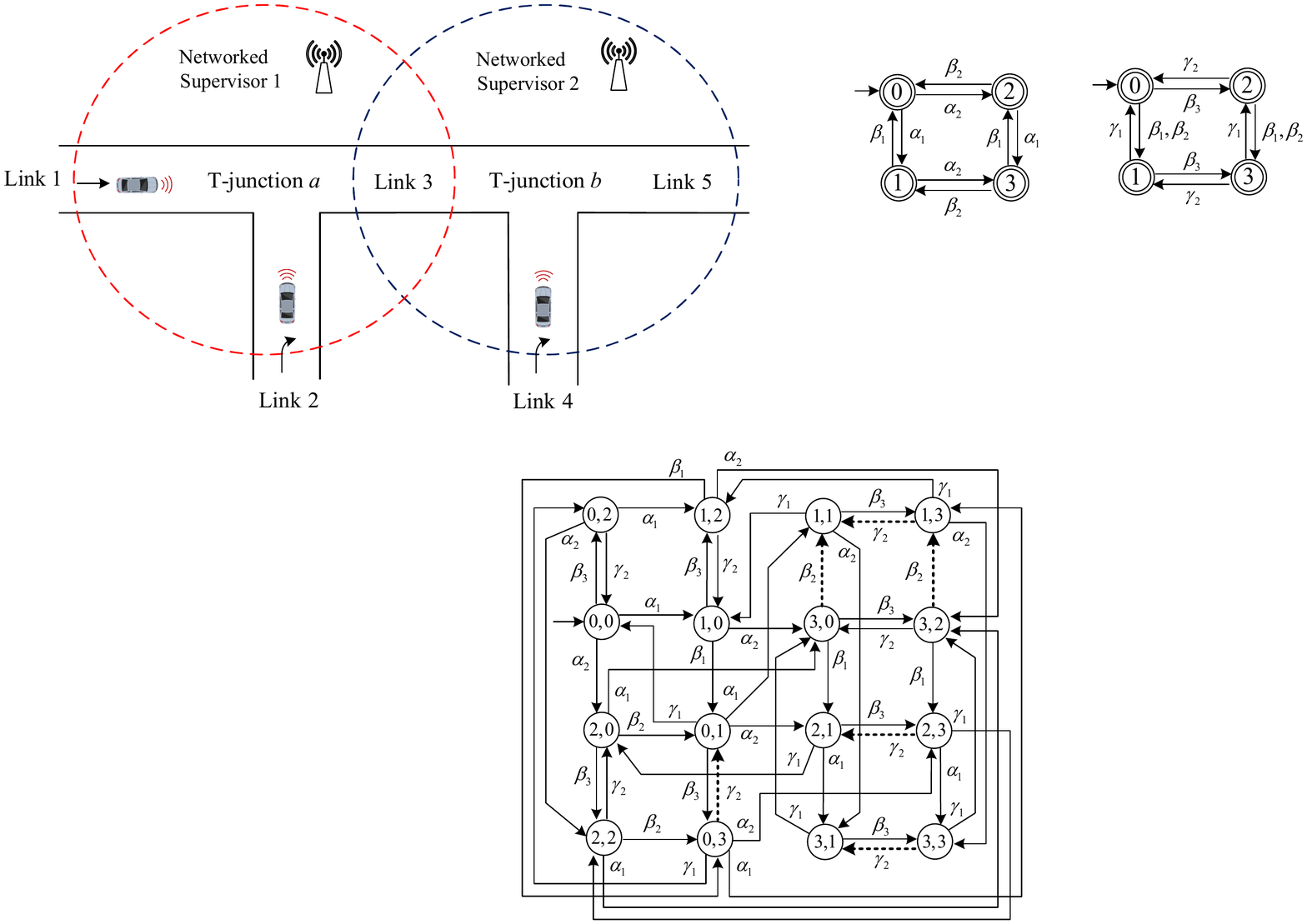}}
\caption{Automata $G_a$ and $G_b$.} \label{Fig6}
\end{figure}

We interpret the construction of $G_a$ as follows.
In the initial state $0$ of $G_a$, vehicles coming from Links 1 or 2 may approach T-junction $a$.
Therefore, both $\alpha_1$ and $\alpha_2$ are defined at state $0$.

When $\alpha_1$ occurs in state $0$,  a vehicle $x$ coming from Link 1 arrives at T-junction $a$, and the system moves to state $1$.
In state 1, vehicle $x$ may pass through T-junction $a$, or another vehicle $y$ coming from Link 2 approaches  T-junction $a$.
Therefore, $\beta_1$ and $\alpha_2$ are defined at state $1$. 
If $\beta_1$ occurs in state 1, vehicle $x$ passes through T-junction $a$, and the system returns to the initial state $0$. 
If $\alpha_2$ occurs in state 1, a new vehicle $y$ coming from Link 2 arrives at T-junction $a$, and the system moves to state $3$.
In state $3$ of $G_a$, both vehicles $x$ and $y$ are waiting to leave T-junction $a$.
Hence, both $\beta_1$ and $\beta_2$ are defined at state $3$.
Upon the occurrence of $\beta_1$ and $\beta_2$ at state $3$, the system moves to states $2$ and $1$, respectively.

When $\alpha_2$ occurs in state $0$, a vehicle $y$ coming from Link 2 arrives at T-junction $a$, and the system moves to state 2.
In state 2, vehicle $y$ may pass through T-junction $a$, or another vehicle coming from Link 1 approaches T-junction $a$.
Hence, both $\beta_2$ and $\alpha_1$ are defined at state $2$.
If $\beta_2$ occurs in state 2, vehicle $y$ leaves T-junction $a$, and the system returns to state $0$. 
If $\alpha_1$ occurs in state 2, a new vehicle $x$ coming from Link 1 approaches T-junction $a$, and the system moves to state $3$ and makes state transitions as mentioned above.

Next, we interpret the construction of $G_b$.
In state 0 of $G_b$, a vehicle coming from Links 1,  2, or  4 may approach T-junction $b$.
Hence, $\beta_1$, $\beta_2$, and $\beta_3$ are defined at state $0$.

When $\beta_1$ or $\beta_2$ occur in state 0, a vehicle $x$ coming from Links 1 or 2  drives into Link 3 and approaches T-junction $b$, and the system moves to state 1.
In state 1 of $G_b$,  vehicle $x$ may pass through T-junction $b$, or vehicle $y$ coming from Link 4 arrives at T-junction $b$.
Correspondingly,  $\gamma_1$ and $\beta_3$ are defined at state 1 in $G_b$.
Upon the occurrence of $\gamma_1$ and $\beta_3$ at state 1, the system moves to states 0 and 3, respectively.
In state 3 of $G_b$, both vehicles $x$ and $y$ can pass through T-junction $b$.
Hence, both $\gamma_1$ and $\gamma_2$ are defined at state $3$  in $G_b$.
If $\gamma_1$ ($\gamma_2$) occurs in state 3, vehicle $x$ (vehicle $y$) leaves T-junction $b$, and the system moves to state 2 (state 1).

When $\beta_3$ occurs in state 0, vehicle $y$ coming from Link 4 approaches T-junction $b$, and the system moves to state 2.
In state 2 of $G_b$, vehicle $y$ may pass through T-junction $b$, or a vehicle coming from Links 1 or 2 arrives at T-junction $b$.
Thus, $\gamma_2$, $\beta_1$, and $\beta_2$  are defined at state 2 in $G_b$.
If $\gamma_2$ occurs in state 2,  vehicle $y$ leaves T-junction $b$, and the system returns to the initial state $0$.
Otherwise, if $\beta_1$ or $\beta_2$ occurs in state 2, a new vehicle $x$ coming from Links 1 or 2 approaches T-junction $b$, and the system moves to state $3$ and makes state transitions as mentioned above.

\begin{figure}
	\begin{center}
		\includegraphics[width=7.2cm]{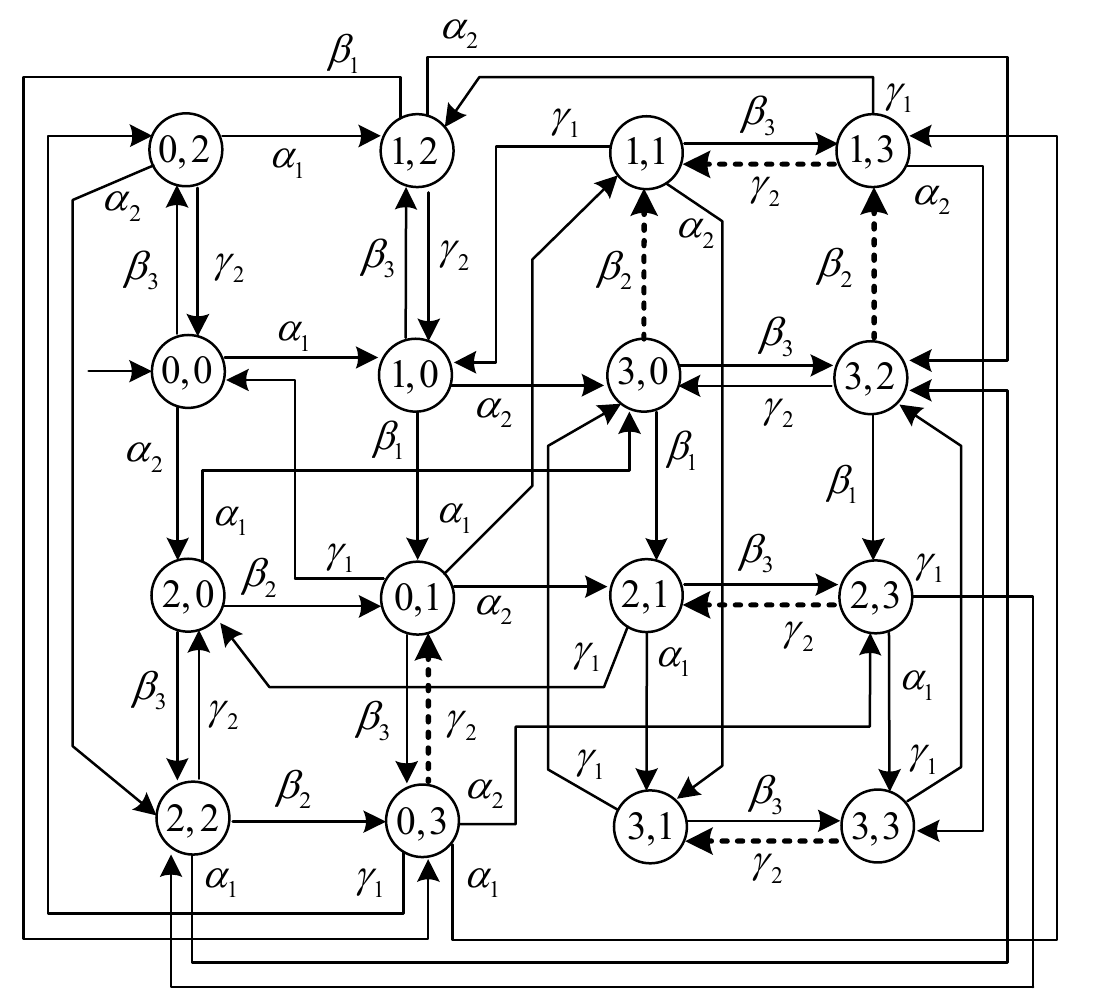}    
		\caption{System model $\tilde{G}$, where all the illegal transitions are highlighted by the dashed line. The desired system $\tilde{H}$ can be obtained from $\tilde{G}$ by deleting all these illegal transitions.} 
		\label{Fig7}
	\end{center}
\end{figure}

As depicted in Fig.\ref{Fig7}, the uncontrolled system $\tilde{G}={G}_a||{G}_b$ models the parallel composition of $G_a$ and $G_b$.
Note that all the states in $G$ are marked states.
In this example, we assume that vehicles driving in the main street have priority in passing through the T-junctions.
That is, if a vehicle coming from Link 1 and a vehicle coming from Link 2 arrive at T-junction $a$ at the same time,  the vehicle coming from Link 2 cannot pass through T-junction $a$ until the vehicle coming from Link 1 leaves T-junction $a$.
Similarly, if a vehicle coming from Links 1 or  2 and a vehicle coming from Link 4 arrive at T-junction $b$ together,  the vehicle coming from Link 4  cannot pass through T-junction $b$ until the  vehicle coming from Links 1 or 2 leaves T-junction $b$.
In  Fig.\ref{Fig7}, all the illegal transitions are highlighted by the dashed line.
As we can see, we disable $\beta_2$ at states (3,0) and (3,2) and disable $\gamma_2$ at states (0,3), (1,3), (2,3), and $(3,3)$. 
The desired system $\tilde{H}\sqsubseteq \tilde{G}$ can be obtained from $\tilde{G}$ by deleting all these illegal transitions.

\subsection{Verification of delay coobservability}

Since Networked Supervisors 1 and  2 are distributed in different T-junctions, we adopt the decentralized control framework as described in Section \ref{Nobs}.
We assume that $\Sigma_{o,1}=\{\alpha_1,\alpha_2,\beta_1,\beta_2\}$ and $\Sigma_{o,2}=\{\beta_1,\beta_2,\beta_3,\gamma_1,\gamma_2\}$.
Meanwhile, we assume that the sets of controllable events for Networked Supervisors 1 and  2 are $\Sigma_{c,1}=\Sigma_{c,2}=\{\beta_1,\beta_2,\gamma_1,\gamma_2\}$.
Since communications between the vehicles and the networked supervisors are achieved over a shared network, communication delays are unavoidable.
We assume that the observation delays for both Networked Supervisors 1 and 2 are 1, i.e., $N_{o,1}=N_{o,2}=1$.
We also assume that there are no control delays for both Networked Supervisors 1 and 2, i.e., $N_{c,1}=N_{c,2}=0$.
By the work of \cite{xu20deds}, there exists a set of networked supervisors for achieving $\mathcal{L}(\tilde{H})$ with certainty under observation delays and control delays iff  $\mathcal{L}(\tilde{H})$ is controllable, $\mathcal{L}_m(\tilde{G})$-closed, and delay coobservable w.r.t. $N_{o,1}+N_{c,1},N_{o,2}+N_{c,2}$, and $\mathcal{L}(\tilde{G})$.
It is not difficult to verify that $\mathcal{L}(\tilde{H})$ is controllable and $\mathcal{L}_m(\tilde{G})$-closed.

For all $s\sigma \in \mathcal{L}(G)$ with $\sigma \in \{\beta_1,\gamma_1\}$, by Fig.\ref{Fig7}, $s\sigma\in \mathcal{L}(H)$.
By Definition \ref{Def1}, to verify if $\mathcal{L}(\tilde{H})$ is delay coobservable, we only need to verify the controllable events $\beta_2,\gamma_2$ regardless of $\beta_1,\gamma_1$.
As pointed out in the previous sections, the verifications for $\beta_2$ and $\gamma_2$ can be done by first constructing verifiers $V^0_{\beta_2}, V^1_{\beta_2}$ and $V^0_{\gamma_2},V^1_{\gamma_2}$ and then checking if there exists a ``bad'' state in these verifiers.
The constructions of these verifiers are similar to those of $V^0_{\gamma}$, $V^1_{\gamma}$, and $V^2_{\gamma}$ in Examples \ref{Ex2} and \ref{Ex3} and are omitted  for brevity.
It can be finally checked that $\mathcal{L}(\tilde{H})$ is delay coobservable w.r.t. $N_{o,1}+N_{c,1},N_{o,2}+N_{c,2}$, and $\mathcal{L}(\tilde{G})$.
Therefore, the control objective $\tilde{H}$ depicted in  Fig.\ref{Fig7}  can be exactly achieved via the decentralized nonblocking networked supervisory control proposed in \cite{xu20deds}.

\section{\textcolor{blue}{Decentralized networked fault diagnosis}}

In this section, we extend our algorithm to verify delay $K$-codiagnodability, which is an important language property arising in the decentralized fault diagnosis of networked DESs.

\subsection{Delay $K$-codiagnosability}

Fault diagnosis is a crucial task in many practical DESs.
It is about determining if a fault event could have occurred or must have occurred in the string executed by the system  \cite{lin94deds, sampath95tac, yin19tac,cao21tcns}.
In many large-scale networked DESs, the information structure of the system is naturally decentralized, as components of the system are distributed.
Decentralized fault diagnosis,  where several diagnosing agents jointly diagnose the plant based on their own observations, is an efficient way of fulfilling fault diagnosis for these large-scale  networked systems \cite{debouk00deds,debouk03deds, qiu06tsmc, qiu08tsmc, yin19tc,yin16ac,kumar09tase,moreira11tac}.
At the same time, in networked systems, undetermined network delays and losses always exist in the data communication between the system and the local diagnosing agents.
Therefore, the (decentralized) fault diagnosis problem with intermittent and permanent loss of observation or network attacks has recently drawn much attention in the DES community \cite{carvalho12ac,carvalho13ac,carvalho17ac,naoki15ijc, takai21ac, wada21deds, marcos22nahs}.
More recently, the effect of observation delays on decentralized fault diagnosis was addressed in \cite{nunes18deds}.

\begin{figure}
	\begin{center}
		\includegraphics[width=7.5cm]{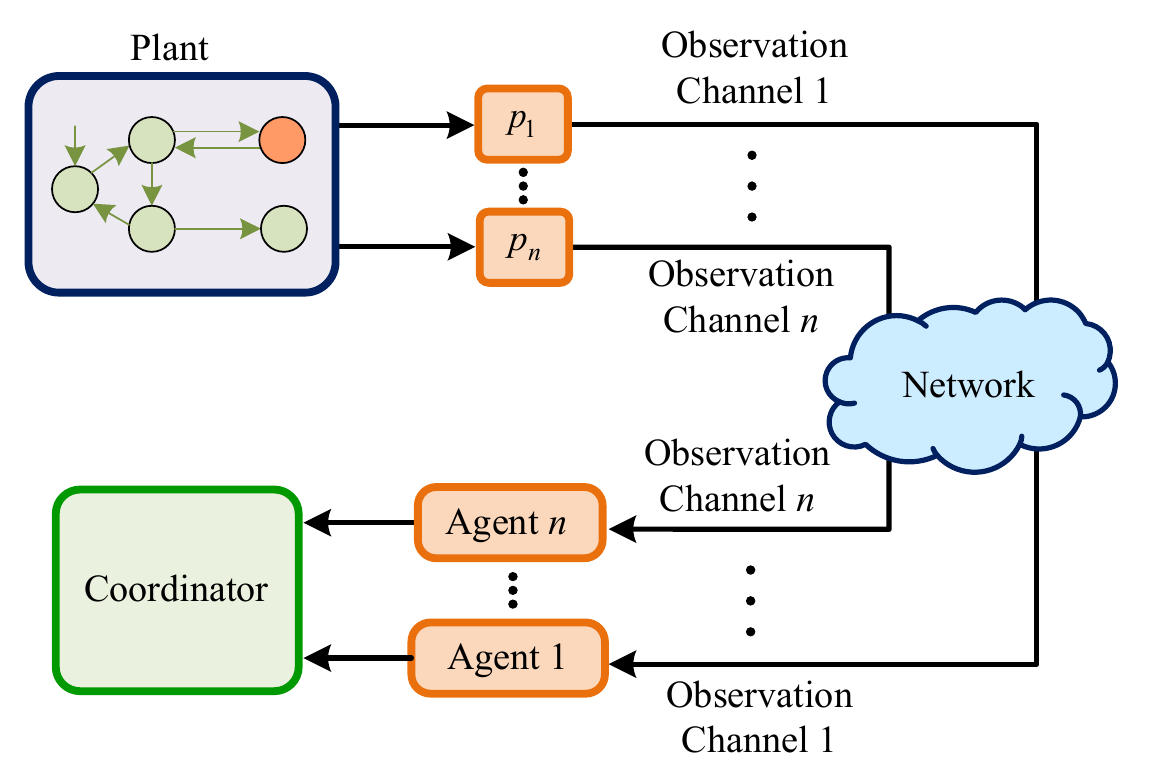}    
		\caption{Decentralized networked fault diagnosis architecture.} 
		\label{Fig9}
	\end{center}
\end{figure}

In this section, we consider the problem of decentralized fault diagnosis under the framework of  networked DESs proposed in \cite{lin2014control,shu14ac, xu20deds}.
Specifically, the adopted decentralized fault diagnosis scheme\footnote{The scheme is similar to Protocol 3 of \cite{debouk00deds} for dealing with the decentralized fault diagnosis problem of non-networked DES.} is depicted in Fig.\ref{Fig9}. 
That is, (i) there is a set of $n$ partial-observation diagnosing agents, each associated with a different projection, jointly diagnosing the system.
(ii) All the local diagnosing agents work independently, i.e., there are no communications among them. 
(iii) The observable event occurrences of the plant are sent to a distributed agent $i\in I$ over an individual observation channel $i$ subject to observation delays.
(iv) The fault event is diagnosed when at least one of the local diagnosing agents identifies its occurrence.
As in \cite{shu14ac, xu20deds}, we assume that
(a) the delays that occur in observation channel $i$ are random but upper bounded by $N_{o,i}$ event occurrences; (b) for each observation channel, delays do not change the order of observations, i.e., FIFO is satisfied.
Thus, for an occurred string $s\in \mathcal{L}(G)$, what the supervisor $i$ may see is $\Theta_i^{N_{o,i}}(s)=\{P_i(s_{-m_i}):m_i\in[0,N_{o,i}]\}$.

\begin{remark}
Some related works are as follows.
In \cite{debouk03deds}, the authors discussed how the communication delays  between the local diagnosing agents and the coordinator can damage the successful fault diagnosis (under Protocols 1 and 2  of \cite{debouk00deds}).
In \cite{qiu08tsmc}, the decentralized fault diagnosis problem was studied under the assumption that (i) the local diagnosing agents could communicate with each other and (ii) communication delays exist among these agents.
Both \cite{debouk03deds} and  \cite{qiu08tsmc} assume that communications between the plant and the local diagnosing agents are reliable.
In practice, however, it is impossible that an event occurrence can be sensed immediately, especially when the local diagnosing agents are far from the plant.
In this paper, we study the decentralized fault diagnosis problem when observation delays exist between the plant and the local diagnosing agents.
\end{remark}

\begin{remark}
We note that the authors in  \cite{nunes18deds} also considered the decentralized fault diagnosis problem under delays between the plant and the local diagnosing agents.
In  \cite{nunes18deds}, it is assumed that (i) each local diagnosing agent observes the system via one or more observation channels (subject to observation delays and losses) and (ii) different observation channels are responsible for different disjoint sets of observable events.
In contrast to  \cite{nunes18deds}, we focus on the case in which each local diagnosing agent communicates with the plant over an individual observation channel.
No single observation channel is uniquely responsible for a set of observable events.
Rather, an observable event occurrence may be sent to several agents over their individual observation channels.
\end{remark}

We denote by $E_f \subseteq \Sigma$  the set of fault events to be diagnosed.
The set of fault events is partitioned into $m$ mutually disjoint sets or fault types: $E_f = E_{f_1}\dot{\cup}\ldots \dot{\cup}E_{f_m}.$
We denote by $\mathcal{F}=\{1,\ldots,m\}$ the index set of the fault types.
Hereafter, ``a fault event of type $f_k$ has occurred" means that a fault event in $E_{f_k}$ has occurred.
Let $\Psi(f_k)=\{t\sigma_f \in \mathcal{L}(G): \sigma_f \in E_{f_k}\}$ be the set of all $s \in \mathcal{L}(G)$ whose last event is a fault event of type $f_k$. 
$\mathcal{L}(G)\setminus s=\{t\in \Sigma^*:st\in\mathcal{L}(G)\}$ denotes the postlanguage of $\mathcal{L}(G)$ after $s$.
We write $f_k\in s$ if $\overline{\{s\}} \cap \Psi(f_k)\neq \emptyset$, i.e., $s$ contains a fault event of type $f_k$. 

We next formalize the notion of delay $K$-codiagnosability.
It is an extension of $K$-codiagnosability\footnote{The formal definition of $K$-codiagnosability  was provided in \cite{debouk00deds,qiu06tsmc}.}  for networked DESs.
A language $L \subseteq \Sigma^*$ is said to be live if whenever $s \in L$, then there exists $\sigma \in \Sigma$ such that $s\sigma \in L$.
We assume that $\mathcal{L}(G)$ is live when delay $K$-codiagnosability is considered.

\begin{definition}
A prefix-closed and live language $\mathcal{L}(G)$ is said to be delay $K$-codiagnosable w.r.t. ${N_{o,1}},\ldots, {N_{o,n}}$, and $E_f$ if the following condition holds:

\begin{align}
&(\forall k \in \mathcal{F})(\forall s \in \Psi(f_k))(\forall t \in \mathcal{L}(G)\setminus s)|t|\ge K  \Rightarrow DCD  \notag\\
&\text{where the delay $K$-codiagnosability condition $DCD$ is}\notag \\
&(\exists i \in I)(\forall u \in \mathcal{L}(G))\Theta_i^{N_{o,i}}(st)\cap \Theta_i^{N_{o,i}}(u) \neq \emptyset \Rightarrow f_k \in u.
\end{align}
\end{definition}

The above definition of delay $K$-codiagnosability states that for any string in the system that contains any type of fault event, at least one local  agent can distinguish that string from strings without that type of fault event within $K$ steps in the presence of observation delays.
Delay $K$-codiagnosability explicitly specifies a uniform detection delay bound for all fault event occurrences.
In other words, delay $K$-codiagnosability requires that any failure must be diagnosed within $K$ steps after its occurrence. 
When $N_{o,i}=0$ for all $i\in I$, delay $K$-codiagnosability reduces to $K$-codiagnosability.

\subsection{The transformation}

When there are no observation delays, it has been shown in \cite{yin15ac} that $K$-codiagnosability is transformable to coobservability.
In this section, we consider the relationship between delay $K$-codiagnosability and delay coobservability.
We show that delay $K$-codiagnosability is also transformable to delay coobservability.
Given automaton $G$, Algorithm 1 constructs two automata $\tilde{G}=(\tilde{Q},\tilde{\Sigma},\tilde{\delta},\tilde{q}_0)$ and $\tilde{H}=(\tilde{Q}_H, \tilde{\Sigma},\tilde{\delta}_H,\tilde{q}_{0,H})$ with $\tilde{H} \sqsubseteq \tilde{G}$ such that the delay $K$-codiagnosability of $G$ is equivalent to the delay coobservability of $\tilde{H}$ w.r.t.  $\tilde{G}$.

\begin{algorithm}[htb]

	\LinesNumbered
		\KwIn{$G=(Q,\Sigma,\delta,q_0)$, $E_f$, and $\mathcal{F}=\{1,\ldots,m\}$.}
		\KwOut{$\tilde{G}=(\tilde{Q},\tilde{\Sigma},\tilde{\delta},\tilde{q}_0)$ and $\tilde{H}=(\tilde{Q}_H, \tilde{\Sigma},\tilde{\delta}_H,\tilde{q}_{0,H})$.} 
          For each  $k \in \mathcal{F}$, build $H_k=(Q_k,\Sigma_k,\delta_k,q_{0,k})$, where $Q_k=Q \times\{-1,0,1,\ldots,K\}$ is the set of states; $\Sigma_k=\Sigma$ is the set of events; $q_{0,k}=(q_0,-1)$ is the initial state, and $\delta_k:Q_k\times \Sigma_k \rightarrow Q_k$ is the transition function where for any $x=(q,n)\in Q_k$ and any $\sigma \in \Sigma_k$, we have
\lfteqn
\delta_k(x,\sigma)=
\begin{cases}
	(\delta(x,\sigma),-1) & \mathrm{if}\ n=-1 \wedge \sigma \in \Sigma \setminus E_{f_k} \\
	(\delta(x,\sigma),n+1) & \mathrm{if}\ [0 \le n <K]  \lor \\&[n=-1 \wedge \sigma \in E_{f_k}]\\
	(\delta(x,\sigma),K) & \mathrm{if}\ n=K;
\end{cases}
\ndeqn
\\
 Set $\hat{H}\leftarrow H_1||H_2||\cdots||H_m$\;
Set $\tilde{H} \leftarrow \hat{H}$, $\tilde{Q}_H\leftarrow \tilde{Q}_H \cup \{legal\}$, and $\tilde{\Sigma}_H=\Sigma \cup \Sigma'$ where $\Sigma'=\{\sigma_k: k \in \mathcal{F}\}$\;
For all $\tilde{q}=(x_1,\ldots, x_m)\in \tilde{Q}_H$ and all $k\in \mathcal{F}$, if $x_k=(q,n)$ with $n=-1$, set $\tilde{\delta}_H(\tilde{q},\sigma_k)=legal$\;
Set $\tilde{G} \leftarrow \tilde{H}$ and $\tilde{Q} \leftarrow \tilde{Q} \cup \{illegal\}$\;
For all $\tilde{q}=(x_1,\ldots, x_m)\in \tilde{Q}$ and all $k\in \mathcal{F}$, if $x_k=(q,n)$ with $n=K$, set $\tilde{\delta}(\tilde{q},\sigma_k)=illegal$\;
		\vspace{.2cm} 		
		\caption{{\sc Transformation} \label{O2D}} 
	\end{algorithm}

We consider observations of the outputs $\tilde{H}$ and $\tilde{G}$.  
For all $i\in I$, let the set of controllable events be $\tilde{\Sigma}_{c,i}=\Sigma_{c,i} \cup \Sigma'$, where $\Sigma'=\{\sigma_k: k \in \mathcal{F}\}$, and the set of observable events be $\tilde{\Sigma}_{o,i}=\Sigma_{o,i}$. 
The natural projection $\tilde{P}_i:\mathcal{L}(\tilde{G})\rightarrow \tilde{\Sigma}_o^*$ is defined for $\tilde{G}$ as $\tilde{P}_i(\varepsilon)=\varepsilon$, and for all $s,s\sigma \in \mathcal{L}(\tilde{G})$, $\tilde{P}_i(s\sigma)=\tilde{P}_i(s)\sigma$ if $\sigma \in \tilde{\Sigma}_{o,i}$, and  $\tilde{P}_i(s\sigma)=\tilde{P}_i(s)$ if $\sigma \in \tilde{\Sigma}\setminus \tilde{\Sigma}_{o,i}$.
We define the observation mapping $\tilde{\Theta}^{N_{o,i}}_{i}:\mathcal{L}(\tilde{G})\rightarrow 2^{\tilde{\Sigma}_o^*}$  for $\tilde{G}$ in the same way as we define $\Theta_{i}^{N_{o,i}}$ for $G$ in Section \ref{Nobs}.

Algorithm 1 is similar to the algorithms proposed in \cite{yin15ac} for transforming  $K$-codiagnosability  to coobservability.
In Line 1, $H_k$ refines the structure of $G$ such that each state of $H_k$ records the number of event occurrences since the occurrence of a fault event of type $f_k$.
In Line 4, we connect $\tilde{q}$ and $legal$ using $\sigma_k$ if strings ending up at $\tilde{q}$ do not contain any fault event of type $f_k$.
In Line 6,  we connect $\tilde{q}$ and $illegal$ using $\sigma_k$ if a fault event of type $f_k$ has occurred $K$ steps earlier.
More construction details on $\tilde{H}$ and $\tilde{G}$ can be found in \cite{yin15ac}.



Assume that $\tilde{G}$ is the uncontrolled system and $\tilde{H}$ is the desired system. 
If there exists $s,s_1,\ldots,s_n \in \mathcal{L}(\tilde{H})$ causing a violation of  delay coobservable w.r.t.  $\mathcal{L}(\tilde{G})$, then $s,s_1,\ldots,s_n$ causes a violation of  delay $K$-codiagnosability, and vice versa, which leads to the following theorem.

\begin{theorem}\label{Theo3}
$\mathcal{L}(G)$ is delay $K$-codiagnosable w.r.t. $N_{o,1},\ldots$, $N_{o,n}$, and $E_f$, iff  $\mathcal{L}(\tilde{H})$ is delay coobservable w.r.t. ${N_{o,1}},\ldots$, $N_{o,n}$, and  $\mathcal{L}(\tilde{G})$.
\end{theorem}
\begin{proof}
($\Rightarrow$)
The proof is by contradiction. 
Assume that $\mathcal{L}(\tilde{H})$ is not delay coobservable  w.r.t. ${N_{o,1}},\ldots, N_{o,n}$, and  $\mathcal{L}(\tilde{G})$.
Note that, for any $\sigma \in \Sigma$ and any $s\in \mathcal{L}(\tilde{G})$, $s\sigma\in \mathcal{L}(\tilde{G})\Leftrightarrow s\sigma \in \mathcal{L}(\tilde{H})$, which implies that (\ref{Eq1}) always holds for $\sigma\in \Sigma$ regardless of $\Theta_i^{N_{o,i}}$.
Therefore, that $\mathcal{L}(\tilde{H})$ is not delay coobservable implies
$(\exists s \in \mathcal{L}(\tilde{H}))(\exists \sigma_k \in \Sigma') [s\sigma_k\in \mathcal{L}(\tilde{G}) \setminus \mathcal{L}(\tilde{H})] \wedge [(\forall i \in I)(\tilde{\Theta}_i^{N_{o,i}})^{-1}(\tilde{\Theta}_i^{N_{o,i}}(s))\sigma_k \cap \mathcal{L}(\tilde{H})\neq\emptyset].
$
Since $s\sigma_k\in \mathcal{L}(\tilde{G}) \setminus \mathcal{L}(\tilde{H})$, by Lines 4 and 6, we can write $s=s_1s_2$ such that $s_1\in \Psi(f_k)$ and $|s_2|\ge K$.
Meanwhile, since  $(\tilde{\Theta}_i^{N_{o,i}})^{-1}(\tilde{\Theta}_i^{N_{o,i}}(s))\sigma_k \cap \mathcal{L}(\tilde{H})\neq\emptyset$, by definition, there exists $t \in \mathcal{L}(\tilde{H})$ such that $\tilde{\Theta}_i^{N_{o,i}}(s) \cap \tilde{\Theta}_i^{N_{o,i}}(t) \neq \emptyset$ and $t \sigma_k \in \mathcal{L}(\tilde{H})$.
By Line 4, $f_k \notin t$.
Note that strings $s$ and $t$ are also in $\mathcal{L}(G)$. 
By the definitions of $\tilde{\Theta}_i^{N_{o,i}}$ and $\Theta_i^{N_{o,i}}$, we have $\tilde{\Theta}_i^{N_{o,i}}(s)={\Theta}_i^{N_{o,i}}(s)$ and $\tilde{\Theta}_i^{N_{o,i}}(t)={\Theta}_i^{N_{o,i}}(t)$.
Hence, ${\Theta}_i^{N_{o,i}}(s) \cap {\Theta}_i^{N_{o,i}}(t) \neq \emptyset$.
Overall,
$(\exists s_1 \in \Psi(f_k))(\exists s_2 \in \mathcal{L}(G)\setminus s_1: |s_2|\ge K)(\forall i \in I)(\exists t \in \mathcal{L}(G)){\Theta}_i^{N_{o,i}}(t) \cap {\Theta}_i^{N_{o,i}}(s_1s_2) \neq \emptyset \wedge f_k \notin t,
$
which contradicts that $\mathcal{L}(G)$ is delay $K$-codiagnosable.

($\Leftarrow$) Also by contradiction. 
Suppose that $\mathcal{L}(G)$ is not delay $K$-codiagnosable w.r.t. ${N_{o,1}},\ldots, N_{o,n}$, and $E_f$.
This implies that
$(\exists f_k \in \mathcal{F})(\exists s \in \Psi(f_k))(\exists t \in \mathcal{L}(G)\setminus s: |t|\ge K) (\forall i \in I)(\exists u \in \mathcal{L}(G)){\Theta}_i^{N_{o,i}}(u) \cap {\Theta}_i^{N_{o,i}}(st) \neq \emptyset \wedge f_k \notin u.
$
Since $s \in \Psi(f_k)$, $|t|\ge K$, and $f_k \notin u$, by Line 1, $\delta_k(q_{0,k},st)=(\delta(q_0,st),K)$ and $\delta_k(q_{0,k},u)=(\delta(q_0,u),-1)$.
By Lines 4 and 6, $st \sigma_k\in \mathcal{L}(\tilde{G})\setminus \mathcal{L}(\tilde{H})$ and $u \sigma_k \in \mathcal{L}(\tilde{H})$.
Moreover, since ${\Theta}_i^{N_{o,i}}(u) \cap {\Theta}_i^{N_{o,i}}(st) \neq \emptyset$, $(\tilde{\Theta}_i^{N_{o,i}})^{-1}(\tilde{\Theta}_i^{N_{o,i}}(s))\sigma_k \cap \mathcal{L}(\tilde{H})\neq\emptyset$.
Overall, we have
$(\exists st \in \mathcal{L}(\tilde{H}))(\exists \sigma_k \in \tilde{\Sigma}) [s\sigma_k\in \mathcal{L}(\tilde{G}) \setminus \mathcal{L}(\tilde{H})] \wedge [(\forall i \in I)(\tilde{\Theta}_i^{N_{o,i}})^{-1}(\tilde{\Theta}_i^{N_{o,i}}(st))\sigma_k \cap \mathcal{L}(\tilde{H})\neq\emptyset],
$
which contradicts that $\mathcal{L}(\tilde{H})$ is delay coobservable.
\end{proof}

\begin{remark}
By Theorem \ref{Theo3}, we can verify if $\mathcal{L}(G)$ is delay $K$-codiagnosable by verifying if $\mathcal{L}(\tilde{H})$ is delay coobservable w.r.t. $\mathcal{L}(\tilde{G})$ using the procedures developed in Section III.
By Lines 1 and 2, $\hat{H}$ has at most $K^m\times|Q|$ states as discussed in \cite{yin15ac}.
By Line 3, the number of states of $\tilde{H}$ is upper bounded by $K^m\times|Q|+1$, and the number of events of $\tilde{H}$ is upper bounded by $|\Sigma|+m$.
Therefore, the worst-case complexity of Algorithm 1 is $\mathcal{O}(K^{m} \times |{Q}| \times |{\Sigma}|)$.
\end{remark}

\begin{remark}
The transformation that we present exploits the fact that problems concerning decentralized state disambiguation (under the framework of networked DESs proposed in \cite{lin2014control,shu14ac, xu20deds}) could be reduced to the delay coobservability problem.
This fact makes our approach for the verification of delay coobservability very useful in other decentralized state disambiguation problems (not confined to just the decentralized fault diagnosis problem), such as, decentralized fault prognosis problem and decentralized detectability problem, in the networked DESs.
\end{remark}

\section{Conclusion} \label{Concl}

In this paper, we revisit the verification of delay coobservability.
A new algorithm for the verification of delay coobservability is proposed.
The algorithm partitions the specification language into different sets of languages.
For each of the sets, a verifier is constructed to check if there exists a string in it causing a violation of delay coobservability.
The complexity of the proposed approach is $\mathcal{O}((n+1) \times |Q_H|^{n+1}\times |\Sigma|^2  \times \sum_{k=0}^{\max\{N_{o,1},\ldots,N_{o,n}\}} (k+1))$ compared with  $\mathcal{O}((n+1)\times |Q_H|^{n+1}\times|\Sigma|^{\max\{N_{o,1},\ldots,N_{o,n}\}+1} \times\prod_{i=1}^n(N_{o,i}+1))$ in \cite{xu20tcns}.
We further consider the decentralized fault diagnosis problem of networked DESs, and introduce the notion of delay $K$-codiagnosability.
We show that delay $K$-codiagnosability is transformable to delay coobservability.
Thus, the procedures developed in this paper for the verification of delay coobservability can also be used to verify delay $K$-codiagnosability.

\appendix

\subsection{Proof of Claim 1}

\begin{proof}
The proof is by induction on $s^d$ for $d=0,1,\ldots,k$.

The base case is for $s^d$ with $d=0$.
Let $t_i$ be the longest prefix of $s_i$ with $P_i(t_i)=P_i(s^0)=\varepsilon$ for all $i\in I^c(\sigma)$.
We write $t_i=\sigma_i^1 \ldots \sigma_i^{k_i}$ for some $\sigma_i^j \in \Sigma_{uo,i}$, $j=1,\ldots,k_i$.
We also write $q_i^j=\delta_H(q_0,t_i^j)$ for $j=0,1,\ldots, k_i$.
By recursively applying (\ref{Eq5-2}), $f_{\sigma}^k(x_{\sigma,0}^k,e_1^1\cdots e_1^{k_1}\cdots e_l^1\cdots e_l^{k_l})=(q_0,q_1^{k_1},\ldots, q_l^{k_l},0)$, where $e_i^j=(\varepsilon,\ldots, \varepsilon, \underset{(i+1)^{st}}{\sigma_{i}^j}, \varepsilon, \ldots,\varepsilon)$ for $i=1,\ldots,l$ and $j=0,1,\ldots,k_i$.

For all $i \in  I^c(\sigma)$ with  $d \ge d_i$, since $d=0$ and $P_i(s_i)=P_i(s^{d_i})$, we have $P_i(s_i)=P_i(s^0)=\varepsilon$.
Moreover, since $t_i$ is the longest prefix of $s_i$ with $P_i(t_i)=P_i(s^d)=\varepsilon$,   $t_i=s_i$.
On the other hand, for all $i \in  I^c(\sigma)$ with  $d< d_i$, we know $q_i^{k_i}=\delta_H(q_0,t_i)$, where $t_i$ is the longest prefix of $s_i$ with $P_i(t_i)=P_i(s^d)$.
Therefore, there exists $(q_0,q_1^{k_1},\ldots, q_l^{k_l},0) \in X^k_{\sigma}$ such that $q_0=\delta_H(q_0,s^d)$ and for all $i \in  I^c(\sigma)$, if $d \ge d_i$, $q_i^{k_i}=\delta_H(q_0,s_i)$, and if $d < d_i$, $q_i^{k_i}=\delta_H(q_0,t_i)$, where $t_i$ is the longest prefix of $s_i$ with $P_i(t_i)=P_i(s^d)$.
The base case is true.

The induction hypothesis is that for all $s^d$ with $d < k$, there exists $(q,q_1,\ldots,q_l,d)\in X^k_{\sigma}$ such that $q=\delta_H(q_0,s^d)$ and for all $i \in  I^c(\sigma)$, if $d \ge d_i$, $q_i=\delta_H(q_0,s_i)$, and if $d < d_i$, $q_i=\delta_H(q_0,t_i)$, where $t_i$ is the longest prefix of $s_i$ with $P_i(t_i)=P_i(s^d)$.
We next prove that the same is also true for  $s^{d+1}=s^d e$.

Let us first consider all the $i \in I^c(\sigma)$ with $d < d_i$.
By the induction hypothesis, $q_i=\delta_H(q_0,t_i)$, where $t_i$ is the longest prefix of $s_i$ with $P_i(t_i)=P_i(s^{d})$.
Since $d< d_i$, $d+1 \le d_i$.
Hence, $s^{d+1}=s^de \in \overline{\{s^{d_i}\}}$.
If $e \in \Sigma_{o,i}$, since $t_i$ is the longest prefix of $s_i$ with $P_i(t_i)=P_i(s^{d})$, we have $t_ie \in \overline{\{s_i\}}$, because otherwise it contradicts $P_i(s_i)=P_i(s^{d_i})$.
By $t_i e \in \overline{\{s_i\}}$, we know $|t_i|<|s_i|$.
Since  $q_i=\delta_H(q_0,t_i)$ and $s_i$ is the shortest string in $\mathcal{L}(H)$ with $P_i(s_i) \in \Theta_i^{N_{o,i}}(s)$ and $s_i\sigma \in \mathcal{L}(H_{N_{o,i}}^{aug})$, $\mathrm{C}_5$ is not satisfied in state $(q,q_1,\ldots,q_l,d)$ for supervisor $i$.
Thus, by $e \in \Sigma_{o,i}$,  $\mathrm{C}_1$ or $\mathrm{C}_3$ may be satisfied in state $(q,q_1,\ldots,q_l,d)$ for supervisor $i$.
On the other hand, if $e \in \Sigma_{uo,i}$, the supervisor $i$ may satisfy $\mathrm{C}_2$ or $\mathrm{C}_4$ or $\mathrm{C}_5$ in state $(q,q_1,\ldots,q_l,d)$.

Let us further consider all the $i \in I^c(\sigma)$ with $d \ge d_i$.
Since $q_i=\delta_H(q_0,s_i)$ and  $s_i\sigma \in  \mathcal{L}(H_{N_{o,i}}^{aug})$, we have $\sigma \in \Gamma_{H,N_{o,i}}^{aug}(q_i)$. 
Since $d_i \in \{\max\{0, k-N_{o,i}\},\ldots, k\}$ and $d \ge d_i$, $k-d \le N_{o,i}$.
Thus, $\textup{C}_5$ is satisfied in state $(q,q_1,\ldots,q_l,d)$ for all $i \in I^c(\sigma)$ with $d \ge d_i$.

Therefore, by (\ref{Eq5-1}), a first type of transition is defined at $(q,q_1,\ldots,q_l,d)$ as:
\begin{align*}
&f_{\sigma}^k((q,q_1,\ldots,q_l,d),(e,e_1,\ldots,e_l))=\\
&(\delta_H(q,e),\delta_H(q_1,e_1),\ldots,\delta_H(q_l,e_l),d+1)=\\
&(q',q'_1,\ldots,q'_l,d+1),
\end{align*}
where for all $i \in I^c(\sigma)$, (i) if $d<d_i \wedge e \in \Sigma_{o,i}$, then $e_i=e$, (ii) if $d<d_i \wedge e \in \Sigma_{uo,i}$, then $e_i=\varepsilon$, and (iii) if $d \ge d_i$, then $e_i=\varepsilon$.

Let $\mathcal{A}=\{m_1,\ldots,m_h\}\subseteq I^c(\sigma)$ be the set of supervisors in $I^c(\sigma)$ such that $d<d_{m_z} \wedge e \in \Sigma_{o,m_z}$ for $z=1,\ldots,h$.
For all $i \in \mathcal{A}$, since $t_i e \in \overline{\{s_i\}}$ and $P_i(t_i e)=P_i(s^d e)$, let $t_i e v_i$ be the longest prefix of $s_i$ with $P_i(t_i e v_i)=P_i(s^d e)$.
We write $v_i=\sigma_i^1 \cdots \sigma_i^{k_i}$ for some $\sigma_i^j \in \Sigma_{uo,i}$, $j=1,\ldots,k_i$.
We also write $q_i^j=\delta_H(q_0,t_i e v_i^j)$ for $j=0,1,\ldots, k_i$.
By recursively applying (\ref{Eq5-2}),  
\begin{align*}
&f_{\sigma}^k((q',q'_1,\ldots,q'_l,d+1),e_{m_1}^1\cdots e_{m_1}^{k_{m_1}}\cdots e_{m_h}^1\cdots e_{m_h}^{k_{m_h}})=\\
&(q',q_1'',\ldots, q_l'',d+1),
\end{align*}
where (i) $e_{m_z}^j=(\varepsilon,\ldots, \varepsilon, \underset{(m_z+1)^{st}}{\sigma_{m_z}^j}, \varepsilon, \ldots,\varepsilon)$ for all $m_z \in \mathcal{A}$ and all $j=1,\ldots,k_{m_z}$, and (ii) $q_{i}''=q_{i}^{k_{i}}$ if $i \in \mathcal{A}$, and $q_i''=q_i'$ if $i \in I^c(\sigma)\setminus \mathcal{A}$.

Then, we have the following results.
1) For all  $i \in I^c(\sigma)$ satisfying $d<d_i  \wedge e \in \Sigma_{o,i}$.
We know $q_i''=q_i^{k_i}=\delta_H(q_0,t_i e v_i)$.
Since $d < d_i$, we have $d+1=d_i$ or $d+1<d_i$.
If $d+1=d_i$, since $t_i e v_i$ is the longest prefix of $s_i$ with $P_i(t_i e v_i)=P_i(s^d e)=P_i(s^{d_i})$, we have  $t_i e v_i=s_i$.
Otherwise, if $d+1<d_i$, $t_i e v_i$ is the longest prefix of $s_i$ with $P_i(t_i e v_i)=P_i(s^d e)=P_i(s^{d+1})$.
2) For all $i \in I^c(\sigma)$ satisfying $d<d_i  \wedge e \in \Sigma_{uo,i}$,  we know $q_i''=q'_i=\delta_H(q_0,t_i)$.
If $d+1=d_i$, since $t_i$ is the longest prefix of $s_i$ with $P_i(t_i)=P_i(s^d)=P_i(s^d e)=P_i(s^{d_i})$, we have $t_i=s_i$.
Otherwise, if $d+1<d_i$, $t_i$ is the longest prefix of $s_i$ with $P_i(t_i)=P_i(s^d)=P_i(s^d e)=P_i(s^{d+1})$.
3) For all $i \in I^c(\sigma)$ satisfying $d \ge d_i$,  we have $d+1 \ge d_i$ and $q_i''=q'_i=\delta_H(q_0,s_i)$.

Therefore, there exists $(q',q_1'',\ldots, q_l'',d+1) \in X^k_{\sigma}$ such that $q'=\delta_H(q_0,s^d e)=\delta_H(q_0,s^{d+1})$ and for all $i  \in I^c(\sigma)$, if $d +1 \ge d_i$, $q''_i=\delta_H(q_0,s_i)$ and if $d+1< d_i$,  $q''_i=\delta_H(q_0,t_i')$, where $t_i'$ is the longest prefix of $s_i$ with $P_i(t_i')=P_i(s^{d+1})$.
That completes the proof.
\end{proof}

\subsection{Proof of Claim 2}

\begin{proof}
The proof is by induction on $s_{-N}w^d$, $d=0,1,\ldots,N$.

The base case is for $s_{-N}w^d$ with $d=0$.
For all $i\in I^c(\sigma)$, since $P_i(u_iw_i)=P_i(s_{-N}w^{d_i})$ and $P_i(u_i)=P_i(s_{-N})$, let $v_i$ be the longest prefix of $w_i$ with $P_i(u_iv_i)=P_i(s_{-N})$. Clearly, $v_i\in \Sigma_{uo,i}^*$.
We write $v_i=\sigma_i^1 \ldots \sigma_i^{k_i}$ for $\sigma_i^j \in \Sigma_{uo,i}$, $j=1,\ldots,k_i$.
We also write $q=\delta_H(q_0,s_{-N})$ and $q_i^j=\delta_H(q_0,u_iv_i^j)$ for $j=0,1,\ldots, k_i$.
By definition, $(q,q_1^0,\ldots, q_l^0) \in x_{\sigma,0}^N$.
By recursively applying (\ref{Eq8}), $f_{\sigma}^N((q,q_1^0,\ldots, q_l^0),e_1^1\cdots e_1^{k_1}\cdots e_l^1\cdots e_l^{k_l})=(q,q_1^{k_1},\ldots, q_l^{k_l},0)$, where $e_i^j=(\varepsilon,\ldots, \varepsilon, \underset{(i+1)^{st}}{\sigma_{i}^j}, \varepsilon, \ldots,\varepsilon)$ for $i=1,\ldots,l$ and $j=0,1,\ldots, k_i$.

For all $i \in  I^c(\sigma)$ with $d \ge d_i$, since $d=0 \ge d_i$, $d_i=0$.
Since $P_i(s_i)=P_i(s_{-N}w^{d_i})$, we have $P_i(s_i)=P_i(s_{-N})$.
Moreover, since $u_iv_i$ is the longest prefix of $s_i$ with $P_i(u_iv_i)=P_i(s_{-N}w^0)$, we have $u_iv_i=u_iw_i=s_i$.
On the other hand, for all $i\in I^c(\sigma)$ with $d < d_i$, we know that $q_i^{k_i}$ tracks the longest prefix $u_iv_i \in \overline{\{s_i\}}$ of $s_i$ such that $P_i(u_iv_i)=P_i(s_{-N}w^0)$.
Thus, there exists $(q,q^{k_1}_1,\ldots,q^{k_l}_l,0)\in X^N_{\sigma}$ such that $q=\delta_H(q_0,s_{-N}w^0)$, and for all $i \in  I^c(\sigma)$, if $d \ge d_i$, $q_i^{k_i}=\delta_H(q_0,s_i)$, and if $d < d_i$, $q_i^{k_i}=\delta_H(q_0,u_iv_i)$, where $u_iv_i$ is the longest prefix of $s_i$ with $P_i(u_iv_i)=P_i(s_{-N}w^0)$.
The base case is true.

The induction hypothesis is that for all $s_{-N}w^d$ with  $d<N$, there exists $(q,q_1,\ldots,q_l,d)\in X^N_{\sigma}$ such that $q=\delta_H(q_0,s_{-N}w^d)$ and for all $i \in  I^c(\sigma)$, if $d \ge d_i$, $q_i=\delta_H(q_0,s_i)$, and if $d < d_i$, $q_i=\delta_H(q_0,u_iv_i)$, where $u_iv_i$ is the longest prefix of $s_i$ with $P_i(u_iv_i)=P_i(s_{-N}w^d)$.
We now prove that the same is also true for $s_{-N}w^{d+1}=s_{-N}w^d e$.

Let us first consider all the $i \in  I^c(\sigma)$ with $d<d_i$. 
By  induction hypothesis,  $q_i=\delta_H(q_0,u_iv_i)$, where $u_iv_i$ is the longest prefix of $s_i$ with $P_i(u_iv_i)=P_i(s_{-N}w^d)$.
Since $d<d_i$,  $d+1 \le d_i$.
Hence, $s_{-N}w^{d+1}=s_{-N}w^de$ is a prefix of $s_{-N}w^{d_i}$.
If $e\in \Sigma_{o,i}$, since $u_iv_i$ is the longest prefix of $s_i$ with $P_i(u_iv_i)=P_i(s_{-N}w^{d})$ and  $P_i(s_i)=P_i(s_{-N}w^{d_i})$, we have $u_iv_ie \in \overline{\{s_i\}}$ and $P_i(u_iv_ie)=P_i(s_{-N}w^{d}e)$.
Since $u_iv_ie \in \overline{\{s_i\}}$,  $|u_iv_i|<|s_i|$.
Moreover, since $q_i=\delta_H(q_0,u_iv_i)$ and $s_i$ is the shortest string in $\mathcal{L}(H)$ such that $P_i(s_i) \in \Theta_i^{N_{o,i}}(st)$ and $s_i\sigma \in \mathcal{L}(H_{N_{o,i}}^{aug})$, $\mathrm{D}_5$ is not satisfied in state $(q,q_1,\ldots,q_l,d)$ for supervisor $i$.
Hence, by $e \in \Sigma_{o,i}$, the supervisor $i$ may satisfy $\mathrm{D}_1$ or $\mathrm{D}_3$ in state $(q,q_1,\ldots,q_l,d)$.
On the other hand, if $e \in \Sigma_{uo,i}$, it can be easily verified that  $\mathrm{D}_2$ or $\mathrm{D}_4$ or $\mathrm{D}_5$ may be satisfied in state $(q,q_1,\ldots,q_l,d)$ for supervisor $i$.

Let us further consider all the $i \in I^c(\sigma)$ with $d \ge d_i$.
Since $q_i=\delta_H(q_0,s_i)$ and  $s_i\sigma \in  \mathcal{L}(H_{N_{o,i}}^{aug})$,  we have  $\sigma \in \Gamma_{H,N_{o,i}}^{aug}(q_i)$. 
Since $d_i \in \{N-N_{o,i},\ldots, N\}$ and $d \ge d_i$, we have $N-d \le N_{o,i}$.
Therefore, $\textup{D}_5$  is satisfied in state $(q,q_1,\ldots,q_l,d)$ for all $i \in  I^c(\sigma)$ with $d \ge d_i$.
Therefore, by (\ref{Eq7}),  a first type of transition is defined at $(q,q_1,\ldots,q_l,d)$ as: 
\begin{align*}
&f_{\sigma}^N((q,q_1,\ldots,q_l,d),(e,e_1,\ldots,e_l))=\\
&(\delta_H(q,e),\delta_H(q_1,e_1),\ldots,\delta_H(q_l,e_l),d+1)=\\
&(q',q_1',\ldots,q'_l,d+1),
\end{align*}
where for all $i \in I^c(\sigma)$, if $d<d_i\wedge e \in \Sigma_{o,i}$, $e_i=e$,  if $d<d_i\wedge e \in \Sigma_{uo,i}$, $e_i=\varepsilon$, and if $d \ge d_i$, $e_i=\varepsilon$.

Let $\mathcal{A}=\{m_1,\ldots, m_h\}\subseteq I^c(\sigma)$ be the set of supervisors such that $d<d_{m_z}$ and $e\in \Sigma_{o,m_z}$ for $z=1,\ldots, h$.
For all $i\in \mathcal{A}$, since $u_i v_i e \in \overline{\{s_i\}}$ and  $P_i(u_i w_i e)=P_i(s^d e)$, let us denote $u_i v_i e z_i$  by the longest prefix of $s_i$ with $P_i(u_i v_i e z_i)=P_i(s^d e)$.
We write $z_i=\sigma_i^1 \ldots \sigma_i^{k_i}$ for $\sigma_i^j \in \Sigma_{uo,i}$, $j=1,\ldots,k_i$.
We also write $q_i^j=\delta_H(q_0,u_i v_i e z_i^j)$ for $j=0,1,\ldots, k_i$.
By recursively applying (\ref{Eq8}),  
\begin{align*}
&f_{\sigma}^N((q',q'_1,\ldots,q'_l,d+1),e_{m_1}^1e_{m_1}^2\cdots e_{m_1}^{k_{m_1}}\cdots e_{m_h}^1e_{m_h}^2\cdots e_{m_h}^{k_{m_h}})=\\
&(q',q_1'',\ldots, q_l'',d+1),
\end{align*}
where (i) $e_{m_z}^j=(\varepsilon,\ldots, \varepsilon, \underset{(m_z+1)^{st}}{\sigma_{m_z}^j}, \varepsilon, \ldots,\varepsilon)$ for all $m_z \in \mathcal{A}$ and all $j=1,\ldots,k_{m_z}$, and (ii) $q_{i}''=q_{i}^{k_{i}}$ if $i \in \mathcal{A}$, and $q_i''=q_i'$ if $i \in I^c(\sigma)\setminus \mathcal{A}$.

Then, we have the following results.
1) For all  $i \in I^c(\sigma)$ satisfying $d<d_i  \wedge e \in \Sigma_{o,i}$, we have $q_i''=q_i^{k_i}=\delta_H(q_0,u_i v_i e z_i)$.
Since $d < d_i$, we have $d+1=d_i$ or $d+1<d_i$.
If $d+1=d_i$, since $u_i v_i e z_i$ is the longest prefix of $s_i$ with $P_i(u_i v_i e z_i)=P_i(s_{-N}w^d e)=P_i(s^{d_i})$, we have $u_i v_i e z_i=s_i$.
Otherwise, if $d+1<d_i$, $u_i v_i e z_i$ is the longest prefix of $s_i$ with $P_i(u_i v_i e z_i)=P_i(s_{-N}w^d e)=P_i(s^{d+1})$.
2) For all $i \in I^c(\sigma)$ satisfying $d<d_i  \wedge e \in \Sigma_{uo,i}$,  we know $q_i''=q'_i=\delta_H(q_0,u_i v_i)$.
If $d+1=d_i$, since $u_iv_i$ is the longest prefix of $s_i$ with $P_i(u_i v_i)=P_i(s_{-N}w^d)=P_i(s_{-N}w^d e)=P_i(s^{d_i})$, we have $u_i v_i=s_i$.
Otherwise, if $d+1<d_i$, $u_i v_i$ is the longest prefix of $s_i$ with $P_i(u_i v_i)=P_i(s_{-N}w^d)=P_i(s_{-N}w^d e)=P_i(s_{-N}w^{d+1})$.
3) For all $i \in I^c(\sigma)$ satisfying $d \ge d_i$,  we have $d+1 \ge d_i$ and $q_i''=q'_i=\delta_H(q_0,s_i)$.

Therefore, there exists $(q',q_1'',\ldots, q_l'',d+1) \in X^N_{\sigma}$ such that $q'=\delta_H(q_0,s_{-N}w^d e)=\delta_H(q_0,s_{-N}w^{d+1})$ and for all $i  \in I^c(\sigma)$, if $d +1 \ge d_i$, $q''_i=\delta_H(q_0,s_i)$ and if $d+1< d_i$,  $q''_i=\delta_H(q_0,u_i v_i')$, where $u_i v_i'$ is the longest prefix of $s_i$ with $P_i(u_i v_i')=P_i(s_{-N}w^{d+1})$.
That completes the proof.
\end{proof}

\bibliographystyle{IEEEtran}     
\bibliography{articles-network}

\end{document}